\documentclass[11pt]{article}
\usepackage{epsfig,color,amsmath,times}
\usepackage{fullpage}
\usepackage{url}
\usepackage{amsfonts}
\usepackage{amsthm}
\usepackage{thmtools} 
\usepackage{thm-restate}

\newcommand{\poly}{{\rm poly}}

  \newtheorem{thm}{Theorem}[section]
 \newtheorem{lemma}[thm]{Lemma}
 \newtheorem{observation}[thm]{Observation}

\newtheorem{definition}{Definition} 
\newtheorem{remark}{Remark}
\newtheorem{fact}{Fact}
\newtheorem{claim}{Claim}

\newcommand\myeq{\mathrel{\overset{\makebox[0pt]{\mbox{\normalfont\tiny\sffamily def}}}{=}}}

\begin{document}

\title{Monotone probability distributions over the Boolean cube can be learned with sublinear samples} 
\begin{titlepage}

\author{
	Ronitt Rubinfeld \thanks{
		CSAIL at MIT, and the
		Blavatnik School of Computer Science at Tel Aviv University,
		{\tt  ronitt@csail.mit.edu}.
		Ronitt Rubinfeld's research was supported by FinTech@CSAIL, MIT-IBM Watson AI Lab and Research Collaboration Agreement No. W1771646, and NSF grants: CCF-1650733, CCF-1740751, CCF-1733808, and IIS-1741137.}
	\and
	Arsen Vasilyan\thanks{	
		CSAIL at MIT,
		{\tt  vasilyan@mit.edu}.
		Arsen Vasilyan's research was supported by the NSF grant IIS-1741137 and EECS SuperUROP program.}
}

\maketitle

\begin{abstract}
	A probability distribution over the Boolean cube is \textbf{monotone} if flipping the value of a coordinate from zero to one can only increase the probability of an element.
	Given samples of an unknown monotone distribution over the Boolean cube,
	we give (to our knowledge) the first algorithm that learns an
	approximation of the distribution in statistical distance
	using a number of samples that is sublinear in the domain.

	To do this, we develop a structural lemma describing monotone
	probability distributions.  The structural lemma has further implications
	to the sample complexity
	of basic testing tasks for analyzing
	monotone probability distributions over the Boolean cube:
	We use it to give nontrivial upper bounds on the tasks of
	estimating the distance of a monotone distribution to uniform and 
	of estimating the support size of a monotone distribution.
	In the setting of monotone probability distributions over the Boolean cube, our algorithms are the first to have sample complexity lower
	than known lower bounds for the same testing tasks on arbitrary (not necessarily monotone) probability distributions. 
	
	One further consequence of our learning algorithm is an improved sample complexity for the task of testing whether a distribution on the Boolean cube is monotone.

\end{abstract}

\end{titlepage}

\section{Introduction}

\subsection{Learning Monotone Distributions}


Data generated from probability distributions is ubiquitous, and 
algorithms for understanding such data are of fundamental
importance.   In particular, a fundamental task is to {\em learn}
an approximation to the probability distribution underlying the
data.
For probability distributions over huge discrete domains,
the sample complexity and run-time bounds for the learning
task can be prohibitive. 
In particular, learning an arbitrary probability distribution on 
a universe of $N_{\text{universe}}$ elements up to sufficiently small constant 
total variation distance requires $\Omega(N_{\text{universe}})$ samples. 
However, when the probability distribution is
known to belong to a more structured class of distributions,
much better results are possible (cf. \cite{DaskalakisKT15, diakonikolas2016efficient, diakonikolas2018fast, diakonikolas2016properly, diakonikolas2016optimal, indyk2012approximating, DaskalakisDS15, DaskalakisDS14, birge1987estimating, chan2013learning, KalaiMV10, GeHK15, DaskalakisK14, 0001V17}) -- for example,
learning an unknown Poisson binomial distribution up to variation distance $\epsilon$ can be achieved with only $\tilde{O}(1/\epsilon^3)$ samples and $\tilde{O}(\log (N_{\text{universe}}) /\epsilon^3)$ run-time \cite{DaskalakisDS15}.


A fundamental class of probability distributions is the class of
multidimensional monotone probability distributions, which broadly satisfy 
the following properties:
\begin{itemize}
	\item The elements of the probability distribution have $n$ different features.
	\item For every element, an increase in the value of one of the features can only increase its probability.
\end{itemize}
This basic class of distributions is of great interest
because  many commonly studied distributions 
are either monotone or can be approximated by a combination of monotone
distributions. 
Furthermore, often the tools developed
for monotone distributions
are useful for other classes of distributions:
for example,
in the one dimensional setting,
\cite{DaskalakisDS14} use tools developed for testing monotone distributions in
order to learn $k$-modal distributions.
In
\cite{canonne2018testing}, tools developed for
testing properties of monotone distributions by \cite{batu2004sublinear}
are used
to develop testers for many other classes of distributions.

For the case of only one feature, or equivalently for monotone probability distributions over the totally ordered set $[k]$, a sample-efficient algorithm is known for learning the unknown distribution up total variation distance $\epsilon$ with $O(\log(k)/\epsilon^3)$ samples \cite{birge1987estimating, DaskalakisDS14}.  
In \cite{AcharyaDK15} it was also shown that an unknown probability distribution over $[k]^n$ can be learned up to $\chi^2$ distance $\epsilon^2$ with $O((n \log k /\epsilon^2)^n/\epsilon^2)$ samples (note that for constant $\epsilon$, this sample complexity is non-trivial only when $k$ is sufficiently large). Overall, the cases considered in the literature specialize on the regime when all the dimensions have a wide range that grows with $n$.
Here we focus on a contrasting case, where each feature has 
only two possible values, $0$ and $1$, thus specializing 
on the Boolean cube:
\begin{definition}
	A probability distribution $\rho$ over $\{0,1\}^n$ is \textbf{monotone} if whenever for $x,y \in \{0,1\}^n$ we have that $x \preceq y$ (which means that for all $i$: $x_i \leq y_i$), then we have that $\rho(x) \leq \rho(y)$.
\end{definition}
When studying multi-dimensional objects, focusing on the specific case of the Boolean cube is a common research theme, because the ideas and techniques developed for the Boolean cube are often applicable in the general case. 
A lower bound of $\Omega(2^{0.15 n})$ for learning monotone probability distributions
over the Boolean cube (up to sufficiently small constant variation distance) can be inferred from 
an entropy testing lower bound in \cite[page 39]{rubinfeldservedio2009testing} and an argument in \cite{valiant2011testing} (see Claim \ref{claim: learning lower bound} in Preliminaries). Though the dramatic exponential improvement as in \cite{birge1987estimating, DaskalakisDS14} for the totally ordered set is thereby impossible, this still leaves open the possibility of a sublinear sample algorithm for the Boolean cube.


We give,
to the best of our knowledge, 
the first sublinear sample algorithm for learning a monotone probability 
distribution over the Boolean cube: 

\begin{restatable}{theorem}{theoremLearningUpperBound}
	\label{theorem: learning upper bound}
	For every positive $\epsilon$, such that $0<\epsilon \leq 1$ and for all sufficiently large $n$, 
	there exists an algorithm, which given $\frac{2^{n}}{2^{\Theta_\epsilon(n^{1/5})}}$ samples from an unknown monotone probability distribution $\rho$ over $\{0,1\}^n$, can reliably return a description of an estimate probability distribution $\hat{\rho}$, such that	$
	d_{\text{TV}}(\rho, \hat{\rho}) \leq \epsilon
	$.
	The algorithm runs in time $O \left(2^{n+O_{\epsilon}(n^{1/5} \log n)} \right)$.
\end{restatable}

Our algorithm relies on a new structural lemma describing monotone
probability distributions on the Boolean cube, as described in Section
\ref{subsection: technical overview}. These structural insights also
allow us to get improved sample complexity for certain testing
tasks on monotone distributions -- namely, 
estimating the closeness of a distribution to uniformity and 
the support size of the distribution, as presented in Section \ref{section: testingpropertiesresults}.

Theorem \ref{theorem: learning upper bound}, together with the $L_1$ distance tester in \cite{valiant2011power}, can be applied to give the best known sample complexity for testing whether a distribution is monotone. Specifically,
one can test whether an unknown distribution $\rho$
over the Boolean cube is monotone or $\epsilon$-far from monotone with $O(\frac{2^n}{ n\epsilon^2})$ samples as shown in Claim \ref{claim: monotonicity testing} in Preliminaries. Note that this does not follow from \cite{valiant2011testing} directly, because monotonicity is not a symmetric property.
The best previously known algorithm for testing monotonicity over the Boolean cube was presented in \cite{bhattacharyya2011testing}, requiring $\tilde{O}\left(\frac{2^n}{(n/\log n)^{1/4}} \poly(1/\epsilon)\right)$ samples. 
The best sample complexity lower bound for testing monotonicity over $\{0,1\}^n$ is $\Omega(2^{(1-\Theta(\sqrt{\epsilon})+o(1)) \cdot n})$, as presented in \cite{AliakbarpourGPR19}. For the domain $[k]^n$, a monotonicity testing algorithm that requires $O\left(k^{n/2}/\epsilon^2+\left(\frac{n \log k}{\epsilon^2} \right)^n \cdot \frac{1}{\epsilon^2}\right)$ samples is given and shown to be optimal in  \cite{AcharyaDK15} (note that this is inapplicable to the Boolean setting, because this sample bound is non-trivial only for sufficiently large $k$). 


\subsection{Testing properties of monotone distributions}
\label{section: testingpropertiesresults}

In addition to learning a distribution, several other 
basic tasks aimed at understanding distributions have
received attention.   These include
estimating the entropy of a distribution, the
size of the support and whether the distribution has
certain ``shape'' properties (monotonicity, convexity,
monotone hazard rate, etc.).
For arbitrary probability distributions over huge domains,
the sample complexity and run-time bounds for the above
tasks can be prohibitive, 
provably requiring $\Omega \left(\frac{N_{\text{universe}}}{\log(N_{\text{universe}})} \right)$ samples. This is true in particular for the properties of support size, entropy and the distance to the uniform distribution  \cite{raskhodnikova2009strong,valiant2011testing, valiant2011estimating, wu2016minimax, wu2019chebyshev}.

This state of affairs motivates going beyond 
worst-case analysis and considering common 
classes of structured probability distributions,  a direction that has
been considered by many and with a large variety of results (cf.
\cite{batu2004sublinear, CanonneDKS17, indyk2012approximating, rubinfeldservedio2009testing, DaskalakisDSVV13, diakonikolas2015testing}).
Some specific examples include:
In \cite{batu2004sublinear} it is shown that testing whether a monotone
distribution is uniform requires only $\Theta(\log^3(N_{\text{universe}})/\epsilon^{3})$ samples,
in contrast to the $\Theta(\sqrt{N_{\text{universe}}}/\epsilon^2)$ samples required for
testing arbitrary distributions for uniformity \cite{paninski2008coincidence,
	ChanDVV14, DiakonikolasGPP19}.
The situation is analogous for the tasks of
testing whether two
distributions given by samples are either the same or far,
and testing whether a constant dimensional distribution is independent, which require only polylogarithmic samples if the unknown distributions are promised to be monotone on a total order \cite{batu2004sublinear}.

Algorithms for testing properties of
monotone probability distributions over the Boolean cube were studied in 
\cite{rubinfeldservedio2009testing, AdamaszekCS10}. 
It was shown that, given samples from a probability distribution over $\{0,1\}^n$ that is promised to be monotone, distinguishing the 
uniform distribution over $\{0,1\}^n$ from one that is $\epsilon$-far from uniform can be done using only $O\left(\frac{n}{\epsilon^2} \right)$ samples, which is nearly optimal.   
In contrast,
a number of other testing problems  cannot have such
dramatic improvements when the distribution is known to be monotone:
for example in \cite{rubinfeldservedio2009testing} it was shown that
for sufficiently small constant $\epsilon$ the estimation of entropy up to an additive error of $\epsilon n$ requires $2^{\Omega(n)}$ samples. However,
no
nontrivial\footnote{i.e. using monotonicity in an essential way and going beyond the bounds known for arbitrary probability distributions.} upper bounds on the sample 
complexity of any other computational tasks for monotone probability distributions over the Boolean cube are known.

\subsubsection{Estimating support size}
We consider the task of additively estimating the support size of an unknown monotone probability distribution over the Boolean cube. 
The following assumption is standard in support size estimation:
\begin{definition}
	A probability distribution over a universe of size $N_{\text{universe}}$ is called \textbf{well-behaved} (in context of support size estimation) if for every $x$ in the set, the probability of $x$ is either zero or at least $1/N_\text{universe}$.
\end{definition}
The purpose of this definition is to rule out pathological cases in which there are items that are in the support, yet have probability very close to zero.
We henceforth adapt this definition to probability distributions over $\{0,1\}^n$, where we have $N_\text{universe}=2^n$. We prove the following theorem:
\begin{restatable}{theorem}{theoremSupportUpperBound}
	\label{theorem: support size upper bound}
	For every positive $\epsilon$, the following is true: for all sufficiently large $n$, 
	there exists an algorithm, which given $\frac{2^{n}}{2^{\Theta_\epsilon(\sqrt{n})}}$ samples from an unknown well-behaved monotone probability distribution $\rho$ over $\{0,1\}^n$, can reliably\footnote{By \textbf{reliably} we henceforth mean that the probability of success is at least $2/3$.} approximate the support size of $\rho$ with an additive error of up to $\epsilon$. The algorithm runs in time $O_{\epsilon}\left(\frac{2^{n}}{2^{\Theta_\epsilon(\sqrt{n})}}\right)$
\end{restatable}
We contrast this result to the results of \cite{raskhodnikova2009strong, valiant2011estimating,valiant2011power,valiant2011testing,wu2019chebyshev} that show that one needs $\Omega(N_{\text{universe}}/\log (N_{\text{universe}}))$ samples to estimate the support size of an arbitrary distribution up to a sufficiently small constant, which equals to $\Omega(2^n/n)$ for a universe of size $2^n$, such as the Boolean cube. 
\subsubsection{Estimating distance to uniformity}
We now consider the task of additively estimating the distance from an unknown monotone probability distribution over the Boolean cube to the uniform distribution. 
We prove the following theorem:
\begin{restatable}{theorem}{theoremDistanceToUniformUpperBound}
	\label{theorem: distance to uniform upper bound}
	For every positive $\epsilon$, the following is true: for all sufficiently large $n$, 
	there exists an algorithm, which given $\frac{2^{n}}{2^{\Theta_\epsilon(\sqrt{n})}}$ samples from an unknown monotone probability distribution $\rho$ over $\{0,1\}^n$, can reliably approximate the distance between $\rho$ and the uniform distribution over $\{0,1\}^n$ with an additive error of up to $\epsilon$. The algorithm runs in time $O \left(2^{n+O_{\epsilon}(\sqrt{n}\log n)} \right)$.
\end{restatable}
We, again, contrast this result to the results of \cite{ valiant2011estimating,valiant2011power,valiant2011testing} that show that one needs $\Omega(N_{\text{universe}}/\log (N_{\text{universe}}))$ samples to estimate the distance of an arbitrary distribution to the uniform distribution, which equals to $\Omega(2^n/n)$ for a universe of size $2^n$, such as the Boolean cube. 

We also have the following sample complexity lower bound on this task, which we prove using the sub-cube decomposition technique of \cite{rubinfeldservedio2009testing}:
\begin{restatable}{theorem}{theoremDistanceToUniformLowerBound}
	\label{theorem: tolerant testing lower bound}
	For infinitely many positive integers $n$,
	there exist two probability distributions $\Delta_{\text{Close}}$ and $\Delta_{\text{Far}}$ over monotone distributions over $\{0,1\}^n$, satisfying:
	\begin{enumerate}
		\item Every distribution in $\Delta_{\text{Far}}$ is $1/2$-far from the uniform distribution.
		\item Any algorithm that takes only $o\left(2^{\frac{n^{0.5-0.01}}{2}}\right)$ samples from a probability distribution, fails to reliably distinguish between $\Delta_{\text{Close}}$ and $\Delta_{\text{Far}}$.
		\item Every distribution in $\Delta_{\text{Close}}$ is $o(1)$-close to the uniform distribution.
	\end{enumerate}
\end{restatable}
\begin{remark} In our construction, the distribution $\Delta_{\text{Close}}$ consists of only one probability distribution. Additionally, the constant $0.01$ can be made arbitrarily small.
\end{remark}

Recall that in
\cite{rubinfeldservedio2009testing, AdamaszekCS10} 
it was shown that, given samples from a probability distribution over the Boolean cube that is promised to be monotone, distinguishing the 
uniform distribution from one that is $\epsilon$-far from uniform can be done using only $O\left(\frac{n}{\epsilon^2} \right)$ samples. Yet, as the theorem above shows, the \textbf{tolerant} version of this problem, which requires one to distinguish a distribution that is $o(1)$-close to the uniform from a distribution that is $1/2$-far from uniform, requires $\Omega\left(2^{\frac{n^{0.5-0.01}}{2}}\right)$ samples, which is dramatically greater.

\subsection{Technical overview}
\label{subsection: technical overview}
\subsubsection{Structural results}
Our analysis applies and builds upon the main structural lemma in \cite{blais2014dnf}. To state it, recall that a \textbf{DNF} is a Boolean function that is formed as an OR of ANDs, and it is monotone if there are no negations. Each AND is referred to as a \textbf{clause}, with the number of variables in the AND is referred to as the \textbf{width} of the clause. Their structural lemma shows that each monotone function can be approximated by a DNF with only a constant number of distinct clause widths. Specifically:
\begin{lemma}[Main Lemma in \cite{blais2014dnf}, abridged and restated]
	\label{lemma BHST}
	For every positive $\epsilon$, for all sufficiently large $n$,
	let $f$ be a monotone Boolean function over the domain $\{0,1\}^n$. There is a function $g=g_1 \lor ... \lor g_t$ with the following properties:
	(i) $t\leq 2/\epsilon$ (ii) each $g_i$ is a monotone DNF with terms of width exactly $k_i$
	(iii) $g$ disagrees with $f$ at no more than $\epsilon \cdot 2^n$ elements of $\{0,1\}^n$ (iv) $g(x) \leq f(x)$ for all $x$ in $\{0,1\}^n$.
\end{lemma}
For Theorem \ref{theorem: support size upper bound}, we use
the lemma above on the indicator function of the support of the probability distribution, which allows us to prove the correctness of our algorithm. For the problems of learning and estimating the distance to uniform, we go a step further and prove an analogous structural lemma for {\em monotone probability distributions}.

There are some crucial differences between monotone Boolean functions in the setting of Boolean function approximation and monotone probability distributions in our setting. First of all, the basic properties of the two objects are different: a Boolean function always has one of the two values (zero or one), which is usually not the case for a probability distribution, but a probability distribution, summed over $\{0,1\}^n$, has to equal one. 
Secondly, the relevant notions of a function $f_2$ being well-approximated by a function $f_1$ are different: for Boolean functions we bound the fraction of points on which $f_1$ and $f_2$ disagree, whereas for monotone probability distributions we would like to bound the $L_1$ distance between $f_1$ and $f_2$.   

To overcome these differences, we generalize to the setting of non-Boolean functions the main concept used in the proof of Lemma \ref{lemma BHST}: the concept of a \textbf{minterm} of a monotone Boolean function. In \cite{blais2014dnf} the minterm of a monotone Boolean function $f$ is defined as follows:
\[
\text{minterm}_f(x) \myeq
\begin{cases}
1 &\text{ if $f(x)=1$ and for all $y \prec x$, $f(y)=0$} \\
0 &\text{otherwise}\\
\end{cases}
\]
Using this language, the function $g$ in Lemma \ref{lemma BHST} can be characterized as a function, for which: 
\begin{equation}
\label{equation: in intro counting with indicators}
\sum_{h=0}^n \mathbf{1}_{\exists \: x \in \{0,1\}^n: \; ||x||=h \: \land \: \text{minterm}_g(x) \neq 0} 
\leq \frac{2}{\epsilon}
\end{equation}

We introduce the notion of monotone \textbf{slack} that generalizes the notion of a minterm to non-Boolean functions:

\[
\text{slack}_f(x)
\myeq
f(x)-\max_{y \prec x} f(y)
=
f(x)-\max_{y \preceq x \text{ and } ||y||=||x||-1} f(y)
\]
With such a definition at hand, one could hope to prove that every monotone probability distribution $\rho$ is well-approximated in the $L_1$ norm by a monotone function $f$, for which $\sum_{h=0}^n \mathbf{1}_{\exists \: x \in \{0,1\}^n: \; ||x||=h \: \land \: \text{slack}_f(x) \neq 0}$ is bounded by a constant independent of $n$.
We were not able to prove such a theorem, and instead we bound a related quantity that can be thought of as the weighted analogue of the expression in Equation \ref{equation: in intro counting with indicators}: $\sum_{h=0}^n R_h \cdot \mathbf{1}_{\exists \: x \in \{0,1\}^n: \; ||x||=h \: \land \: \text{slack}_f(x) \neq 0}$, where the $R_h$ are positive weights that can be chosen arbitrarily, as long as they satisfy a certain technical condition that ensures that not too many of these weights are too large. Precisely, our main lemma is:
\begin{restatable}[Main Structural Lemma]{lemma}{ourMainLemma}
	\label{lemma: slack regret}
	For all positive $\zeta$, for all sufficiently large $n$, the following is true: Let $\rho$ be a monotone probability distribution over $\{0,1\}^n$.
	Suppose, for each $h$ between $0$ and $n$ we are given a positive value $R_h$, and it is the case that:
	$\;\;
	\sum_{h=0}^n R_h \cdot \frac{\binom{n}{h}}{\sum_{j=h}^n \binom{n}{j}}
	\leq \zeta
	$
	
	Then, there exists a positive monotone function $f$, mapping $\{0,1\}^n$ to positive real numbers, satisfying:
	\begin{enumerate}
		\item For all $x$, it is the case that $\rho(x) \geq f(x)$.
		\item It is the case that:
		$\;\;
		\sum_{x \in \{0,1\}^n} \rho(x)-f(x) \leq \zeta
		$
		\item It is the case that:
		$
		\;\;
		\sum_{h=0}^n R_h \cdot \mathbf{1}_{\exists \: x \in \{0,1\}^n: \; ||x||=h \: \land \: \text{slack}_f(x) \neq 0}
		\leq 1
		$
	\end{enumerate}
\end{restatable}

Now, as a corollary, we present a simple special case (proven to be so in Subsection \ref{section: slackness regret}) that not only illustrates the power of Lemma \ref{lemma: slack regret}, but also is sufficient for our proof of Theorem \ref{theorem: distance to uniform upper bound}:  

\begin{restatable}{corollary}{corollaryOfOurMainLemma}
	\label{corollary: slack regularity old lemma}
	Let $\rho$ be a monotone probability distribution over $\{0,1\}^n$ and let $h_0$ be an integer for which:
	\[
	\frac{\epsilon}{4}
	\leq
	\Pr_{x \sim \{0,1\}^n}
	[||x|| \geq h_0]
	\leq
	\frac{\epsilon}{2}
	\]
	Then, there exists a positive monotone function $f:\{0,1\}^n \rightarrow R$ satisfying:
	\begin{enumerate}
		\item For all $x$, it is the case that $\rho(x) \geq f(x)$.
		\item It is the case that:
		$\;\;
		\sum_{x \in \{0,1\}^n} \rho(x)-f(x) \leq \frac{\epsilon}{4}
		$.
		\item There exists a set of values $\{k_1,...,k_t\}$ (ordered in an increasing order) with $t \leq \frac{16}{\epsilon^2}$, satisfying that if for some $x$ in $\{0,1\}^n$ we have $||x|| < h_0$ and $\text{slack}_{f}(x) \neq 0$, then $||x||=k_i$ for some $i$.
	\end{enumerate}
\end{restatable}
For Theorem \ref{theorem: learning upper bound}, however, we use the full power of Lemma \ref{lemma: slack regret}.
\subsubsection{Algorithmic ideas}
Here we present an informal overview of the ideas involved in the design and analysis of our algorithms. Throughout we omit details and technicalities.
As already mentioned, our algorithms for Theorems \ref{theorem: support size upper bound}, \ref{theorem: distance to uniform upper bound} and \ref{theorem: learning upper bound} use respectively Lemma \ref{lemma BHST}, Corollary \ref{corollary: slack regularity old lemma} and Lemma \ref{lemma: slack regret} as their structural core.
Here we present the algorithmic ideas in the order of increasing technical sophistication.
\paragraph{Support size estimation (Theorem \ref{theorem: support size upper bound})}
The idea behind our support size estimation algorithm is as follows: if we received $x$ as a sample, then not only $x$ has to be in the support of $\rho$, but every $y$, satisfying $x \preceq y$ is in the support of $\rho$. For all such $y$, we say that $y$ is \textbf{covered} by $x$. Our algorithms estimates the support size of $\rho$ through estimating the number of all such $y$ that are covered by at least one of the samples. 

This algorithm can be made computationally efficient by standard methods in randomized algorithms, and the only non-trivial step is to show that $\frac{2^n}{2^{\Theta_{\epsilon}(\sqrt{n})}}$ samples suffice. To show this, we first apply Lemma \ref{lemma BHST} to the indicator function of the support of $\rho$ (which we from now on call the support function of $\rho$). This gives us a Boolean function $g$ that approximates well the support function of $\rho$ and has zero slack everywhere, except for a small number of levels\footnote{i.e. subsets of $\{0,1\}^n$ that have the same Hamming weight.} of $\{0,1\}^n$. For simplicity, assume that the support function of $\rho$ itself has this property, and there are only a small number of levels of the Boolean cube on which the support function of $\rho$ can have non-zero slack, which we call the \textbf{slacky} levels.

Now, we divide the elements of $\{0,1\}^n$ (which we also call \textbf{points}) into \textbf{good}\footnote{We later re-define these notions in order to adapt them for the technical details we ignore in the introduction.} points and \textbf{bad} points, with the former defined as all the points sufficiently close to a slacky level, and the latter defined as all the other points. Clearly, a given level of $\{0,1\}^n$ consists either fully from good points or fully from bad points, so we also refer to levels as good or bad.

We argue that if a point $y$ in the support of $\rho$ is a good point, then it is likely to be covered by one of the samples, because there is a large number of values $x$ in the support of $\rho$, for which $y \preceq x$. 

We conclude by bounding the number of elements in the support of $\rho$ that are bad, by using the fact that there cannot be too many slacky levels.
\paragraph{Estimation of distance to uniform (Theorem \ref{theorem: distance to uniform upper bound})}
To estimate the distance of a monotone distribution to uniform, we pick a value $h_0$ as in Corollary \ref{corollary: slack regularity old lemma} and break down the value of the total variation distance from $\rho$ to uniform into contributions from two disjoint components: (i) $\{x \in \{0,1\}^n \text{ s.t. } ||x|| \geq h_0\}$ and (ii) all the other points of $\{0,1\}^n$. In other words, we use $h_0$ as the cutoff value for the Hamming weight, to separate $\{0,1\}^n$ into components (i) and (ii). The first contribution is straightforward to estimate simply through estimating how likely a random sample $x$ from $\rho$ is to have $||x|| \geq h_0$, because it is straightforward to prove that if one redistributes the probability mass of $\rho$ in $\{x \in \{0,1\}^n \text{ s.t. } ||x|| \geq h_0\}$, while keeping the total amount of probability mass in this set fixed, the total variation distance between $\rho$ and the uniform distribution cannot change by more than $O_{\epsilon}(1)$.

For any element $x$ of the component (ii), we prepare an estimate of $\rho(x)$, which we call $\hat{\phi}(x)$. Our approach here is somewhat similar to the one for our support size estimation algorithm. 
In the case of support size estimation, we only registered whether $x$ was covered by a sample from $\rho$ or not. In this case, we actually need an estimate on $\rho(x)$ (as opposed to $\mathbf{1}_{\rho(x)\neq 0}$) which we obtain by studying the pattern of all the samples covering $x$. 
More precisely, suppose we draw $N_2$ samples from the distribution, which form a multiset $S_2$. We extract the estimate $\hat{\phi}(x)$ from the pattern of samples as follows:
\[
\hat{\phi}(x)
:=
\frac{1}{2^L} \cdot \frac{\max_{y \text{ s.t. } y \preceq x \text{ and } ||x||-||y|| = L}
	\bigg \lvert \bigg \{z \in S_2: y \preceq z \preceq x \bigg\} \bigg \rvert}{N_2}
\] 
Here $L$ is a parameter equal to $\Theta_{\epsilon}(\sqrt{n})$. We then estimate the contribution of set (ii) as $\sum_{x \in \{0,1\}^n: \; ||x|| \geq h_0} \left\lvert \hat{\phi}(x)-1/2^n \right\rvert$.

We show the correctness of our algorithm as follows.
We use a tail bound to show that $\hat{\phi}(x)$ concentrates sufficiently closely to the value: .
\begin{multline*}
\phi(x)
\myeq
\frac{1}{2^L} \cdot \max_{y \text{ s.t. } y \preceq x \text{ and } ||x||-||y|| = L} \text{  }
\Pr_{z \sim \rho}[ y \preceq z \preceq x ]
=\\
\frac{1}{2^L} \cdot \max_{y \text{ s.t. } y \preceq x \text{ and } ||x||-||y|| = L} \text{  }
\sum_{z \text{ s.t. } y \preceq z \preceq x} \rho(z)
\end{multline*}
Then, we apply Corollary \ref{corollary: slack regularity old lemma}, which implies that $\rho$ is approximated well by a function $f$ and a certain set of constraints on the slack of $f$ holds. 

Now for the sake of simplicity (analogously to the case of support size estimation), assume that $\rho$ itself satisfies the condition that below the threshold $h_0$ there are at most $O_{\epsilon}(1)$ levels of $\{0,1\}^n$ on which there are points $x$ with non-zero $\text{slack}_{\rho}(x)$ (in reality it is merely well-approximated by such a function).
We now can (analogously to the case of support size estimation) 
introduce the concepts of \textbf{slacky} levels as levels on which $\rho$ has non-zero slack, and \textbf{good} levels, which are below $h_0$ and farther than $L$ from all \textbf{slacky} levels of $\rho$. 
Now, one can prove that
for $x$ on a good level the value of
$\phi(x)$ equals precisely to $\rho(x)$, for the following reasons: First of all the inequality: \[\max_{y \text{ s.t. } y \preceq x \text{ and } ||x||-||y|| = L} \rho(y) \leq \phi(x) \leq \rho(x)\] follows immediately from the monotonicity of $\rho$ and the definition of $\phi$. Secondly, if $\rho$ has no slack on the levels between $||x||$ and $||x||-L$ (inclusive), then from the definition of slack it follows immediately using induction on $L$ that: \[\max_{y \text{ s.t. } y \preceq x \text{ and } ||x||-||y|| = L} \rho(y)=\rho(x)\] Therefore, it has to be the case that $\rho(x)=\phi(x)$.

Finally, we bound the contribution to the $L_1$ distance between $\hat{\phi}$ and $\rho$ of all the levels below $h_0$ that are not good (which we again call the \textbf{bad} levels). We do this by upper-bounding the number of bad levels, and then upper bounding the total probability mass on a single level below $h_0$.
\paragraph{Learning a monotone probability distribution (Theorem \ref{theorem: learning upper bound})}

As we saw, our algorithm for the estimation of the distance to the uniform distribution contained a component that learned in $L_1$ distance the restriction of $\rho$ on the levels below the cutoff $h_0$.
The main challenge here is to extend these ideas to levels above $h_0$. To this, we make the following changes to our setup:
\begin{itemize}
	\item Instead of having one fixed constant $L$ defining whether a point is close to a slacky level, we make this value level-dependent. In other words, for every $h$ we define $L_h$, and then after drawing $N$ samples, which form a multiset $S$, we compute:
	\[
	\hat{\phi}(x)
	:=
	\frac{1}{2^{\left \lfloor L_{||x||} \right \rfloor}} \cdot \frac{\max_{y \text{ s.t. } y \preceq x \text{ and } ||x||-||y|| = \left \lfloor L_{||x||} \right \rfloor}
		\bigg \lvert \bigg \{z \in S: y \preceq z \preceq x \bigg\} \bigg \rvert}{N}
	\]
	\item  Instead of using Corollary \ref{corollary: slack regularity old lemma}, we use the the full power of Lemma \ref{lemma: slack regret}. This, again gives us a function $f$ that approximates $\rho$ closely and has a restriction on its slacky levels.
\end{itemize} 
Finally, we pick values of $L_h$ in the algorithm and $R_h$ in the analysis so we balance (i) The random error from the deviation of $\hat{\phi}(x)$ from its expectation and (ii) The systematic error introduced by the slacky levels of $f$ and the levels close to them. As a result, we find that $\frac{2^n}{2^{\Theta(n^{1/5})}}$ samples suffice.

\section{Preliminaries}

We use the following basic definitions and notation:

\begin{definition}
	For $x \in \{0,1\}^n$, its \textbf{Hamming weight} is denoted as $\lvert \lvert x \rvert \rvert$ and is equal to $\sum_i x_i$.
\end{definition}



\begin{definition}
	\label{definition: average value of a function on a level}
	For a function $f: \{0,1\}^n \rightarrow R$, we define the \textbf{average value on level $k$} (with $0\leq k \leq n$) as:
	$
	\mu_{f}(k)=\frac{1}{\binom{n}{k}}\sum_{x\in \{0,1\}^n: ||x||=k}{f(x)}
	$.
	We also refer to average value on level $k$ for a probability distribution $\rho$, which we denote $\mu_{\rho}(k)$. By this we mean the average value on level $k$ of the density function of $\rho$.  
\end{definition}

\begin{definition}
	\label{definition: monotone slack}
	For a monotone function $f: \{0,1\}^n \rightarrow R$, we define the \textbf{monotone slack} $\text{slack}_f(x)$ at point $x \in \{0,1\}^n$ as follows:	$
	\text{slack}_f(x)
	\myeq
	f(x)-\max_{y \prec x} f(y)
	=
	f(x)-\max_{y \preceq x \text{ and } ||y||=||x||-1} f(y)
	$.
	We also stipulate that $\text{slack}_f(0^n)=f(0^n)$.
\end{definition}

\begin{definition}
	The \textbf{total variation distance} between two probability distributions $\rho_1$ and $\rho_2$ is defined as:
	
	$
	d_{\text{TV}}(\rho_1, \rho_2)
	\myeq \frac{1}{2} \sum_{x} \lvert \rho_1(x) - \rho_2(x)\rvert
	$.
\end{definition}

The following are well-known facts, which were also used in \cite{blais2014dnf}:

\begin{fact}
	\label{fact: average density increases}
	For a monotone function $f: \{0,1\}^n \rightarrow R$, for all $k_1, k_2$ satisfying $0\leq k_1 \leq k_2 \leq n$, it is the case that $\mu_f(k_1) \leq \mu_f(k_2)$.
\end{fact}

\begin{fact}
	\label{fact: boolean cube anti-concentration}
	For all $k$, it is the case that
	$
	\binom{n}{k}
	\leq
	\frac{2}{\sqrt{n}} \cdot 2^n
	$.
\end{fact}
Now, we justify two claims we made in the introduction:
\begin{claim}
	\label{claim: learning lower bound}
	For sufficiently small $\epsilon_0$, for all sufficiently large $n$, any algorithm that learns an unknown monotone probability distribution over $\{0,1\}^n$ requires at least $\Omega(2^{0.15 n})$ samples from the distribution.
\end{claim}
\begin{proof}
	From the argument in \cite[pages 1937-1938]{valiant2011testing} it follows that if two probability distributions are $\epsilon$-close in total variation distance, then their entropy values are within $2 \log(N_{\text{universe}}) \epsilon=2\epsilon n$. Therefore, the task of estimating the entropy of an unknown monotone probability up to an additive error $2 \epsilon n$ is not harder than learning it to withing total variation distance $\epsilon$. But in \cite[page 39]{rubinfeldservedio2009testing}
	it is shown that at least $\sqrt{T}/10$ samples are required for the task of distinguishing whether the unknown monotone probability distribution has entropy at least $0.81n$ or at most $n/2+\log T$. Picking $T=2^{0.3n}$ gives us the desired learning lower bound. 	
\end{proof}

\begin{claim}
	\label{claim: monotonicity testing}
	Given Theorem \ref{theorem: learning upper bound}, one can test whether an unknown distribution $\rho$
	over the Boolean cube is monotone or $\epsilon$-far from monotone with $O(\frac{2^n}{ n\epsilon^2})$ samples. 
\end{claim}
\begin{proof}
	This can be done in the following way:
	(1) Use our learning algorithm with an error parameter $\epsilon/4$. This gives us a description of a distribution $\hat{\rho}$, which is $\epsilon/4$-close to $\rho$ if $\rho$ is monotone.
	(2) Estimate, using the estimator of \cite{valiant2011power}, the total variation distance between $\rho$ and $\hat{\rho}$ up to $\epsilon/4$. If the result is closer to $\epsilon$ than to zero, output NO. (3) Compute the total variation distance between $\hat{\rho}$ and the closest monotone probability distribution. If this distance estimate is closer to $\epsilon$ than to zero, output NO, otherwise output YES. For constant $\epsilon$, the sample complexity is dominated by step (3), which is $O(\frac{2^n}{\epsilon^2 n})$. It is easy to see that a monotone probability distribution will pass this test, whereas a distribution that is $\epsilon$-far from monotone will fail either step (2) or step (3).
\end{proof}
\section{Learning monotone probability distributions}

\begin{figure}
	\caption{}
	\label{algorithm: learning}
			\textbf{Algorithm for learning a monotone probability distribution over the Boolean cube} (given sample access from a distribution $\rho$, which is monotone over $\{0,1\}^n$).
			
			\dotfill

			\begin{enumerate}
				\item Set $A:=\frac{1}{2n} \cdot e^{\frac{1}{2000} \cdot n^{1/5}}$. For all $h \geq n/2$, set $L_h:=\max\left(\log\left(2 n A \cdot \frac{\binom{n}{h}}{2^n} \right
				),0\right)$
				
				Similarly, for all $h$, satisfying $n/2 > h \geq 0$, set:
				$L_h:= L_{n/2}=\log\left(2 n A \cdot \frac{\binom{n}{n/2}}{2^n} \right)$.  
				\item Set 
				$N:= \frac{2^n}{A} \cdot \frac{192}{\epsilon^2} \cdot (n+9 \sqrt{n}+4)
				$
				
				Draw $N$ samples from the probability distribution $\rho$ and denote the multiset of these samples as $S$.
				\item For all $x$ in $\{0,1\}^n$, if $||x|| < 9 \sqrt{n}$, then set $\hat{\phi}(x)=0$, otherwise compute:
				\[
				\hat{\phi}(x)
				:=
				\frac{1}{2^{\left \lfloor L_{||x||} \right \rfloor}} \cdot \frac{\max_{y \text{ s.t. } y \preceq x \text{ and } ||x||-||y|| = \left \lfloor L_{||x||} \right \rfloor}
					\bigg \lvert \bigg \{z \in S: y \preceq z \preceq x \bigg\} \bigg \rvert}{N}
				\] 
				Do this by first making a look-up table, which given arbitrary $z \in \{0,1\}^n$ returns the number of times $z$ was encountered in $S$. Then, use this look-up table to compute the necessary values of $\lvert \{z \in S: y \preceq z \preceq x \} \rvert$  by querying all these values of $z$ in the lookup table and summing the results up. 
				\item For all $x$ in $\{0,1\}^n$, compute the following:
				$
				\hat{\rho}(x)
				=\hat{\phi}(x)+\frac{1}{2^n}\left(1-\sum_{y \in \{0,1\}^n} \hat{\phi}(y)  \right) 
				$
				\item Output the value table of $\hat{\rho}$.
			\end{enumerate}
\dotfill
\end{figure}
In this section we prove our upper-bound on the sample complexity of learning an unknown monotone probability distribution over the Boolean cube. We restate the theorem:
\theoremLearningUpperBound*
\begin{proof}
	We present the algorithm in Figure \ref{algorithm: learning}.
	The number of samples drawn from $\rho$ is $N=\frac{2^{n}}{2^{\Theta_{\epsilon}(n^{1/5})}}$. The run-time, in turn, is dominated by computing the values of $\hat{\phi}$ in step (3), in which the construction of the lookup table takes $O(n \cdot 2^n)$ time, and the time spent computing each $\hat{\phi}(x)$ can be upper bounded by the product of: (i) the number of pairs $(y,z)$ that simultaneously satisfy $y \preceq z \preceq x$ and $||y||-||x|| =L_{||x||}$, which can be upper-bounded by $O(n^{L_{||x||}} \cdot 2^{L_{||x||}})$ and (ii) the time it takes to look up a given $z$ in the lookup table, which can be upper-bounded by $O(n)$. Overall, this gives us a run-time upper bound of $O(2^{n+O_{\epsilon}(n^{1/5} \log n
		)})$.
	
	Now, the only thing to prove is correctness. Here is our main claim:
	
	\begin{claim}
		\label{claim: learning main claim}
		If the following conditions are the case: 
		\begin{itemize}
			\item[a)] As a function of $h$, $L_h$ is non-increasing.
			\item[b)] For all $h$, we have that $L_h \leq 9 \sqrt{n}$.
			\item[c)] \[\frac{1}{2^n}
			\cdot 
			\sum_{h=9 \sqrt{n}}^n
			\binom{n}{h}
			\cdot
			\frac{A}{2^{L_{h}}} \leq \frac{1}{2}\]
			\item[d)] \[\sum_{h=9 \sqrt{n}}^n L_h \cdot \left(\begin{cases}
			\frac{400}{n^{2.5}} &\text{ if $h \leq n/2-\sqrt{n \ln(n)}$} 
			\\ \frac{40000}{n} &\text{ if $n/2-\sqrt{n \ln(n)}< h < n/2+\sqrt{n}$} \\ 40000 \cdot \left(\frac{h-n/2}{n}\right)^2 &\text{ if $h \geq n/2+\sqrt{n}$}\end{cases} \right)\leq \frac{\epsilon^2}{20000}\]
		\end{itemize}
		Then, with probability at least $2/3$, it is the case that
		$\sum_{x \in \{0,1\}^n \text{ s.t. } 9 \sqrt{n} \leq ||x||}
		\bigg \lvert \hat{\phi}(x)-\rho(x) \bigg \rvert
		\leq \frac{\epsilon}{2}$.
	\end{claim}
	
	We verify in Appendix A, subsection \ref{appendix subsection: verifying the conditions on L_h}, that $L_h$ indeed satisfy the conditions above.
	In fact, the values of $L_h$ and $A$ were chosen specifically to satisfy the constraints above.
	We prove Claim \ref{claim: learning main claim} in Section \ref{section: proof lerning main claim}, after we develop our main structural lemma in Section \ref{section: slackness regret}.
	
	We now bound the contribution to the $L_1$ distance between $\hat{\phi}$ to $\rho$ that comes from points of Hamming weight less than $9 \sqrt{n}$. Since $\sum_{x \in \{0,1\}^n} \rho(x)=1$ and $\rho$ is monotone, then whenever $||x|| \leq n/2$ we have $\rho(x) \leq 1/2^{n/2}$. Therefore, for sufficiently large $n$ we have:
	\[
	\sum_{x \in \{0,1\}^n \text{ s.t. }  ||x|| < 9 \sqrt{n} }
	\bigg \lvert \hat{\phi}(x)-\rho(x) \bigg \rvert
	=\sum_{x \in \{0,1\}^n \text{ s.t. }  ||x|| < 9 \sqrt{n} }
	\rho(x)
	\leq
	\frac{n^{9 \sqrt{n}}}{2^{n/2}}
	\leq \frac{\epsilon}{2}
	\]  
	Combining this with the bound in Claim \ref{claim: learning main claim} we get:
	\[
	\sum_{x \in \{0,1\}^n}
	\bigg \lvert \hat{\phi}(x)-\rho(x) \bigg \rvert
	\leq \epsilon
	\]  
	Overall, we have:
	\begin{multline*}
	2 \cdot d_{\text{TV}}(\rho, \hat{\rho})
	=
	\sum_{x \in \{0,1\}^n}
	\bigg \lvert \hat{\rho}(x)-\rho(x) \bigg \rvert
	=\\
	\sum_{x \in \{0,1\}^n}
	\bigg \lvert \hat{\phi}(x)-\rho(x) + \frac{1}{2^n}\left(1-\sum_{y \in \{0,1\}^n} \hat{\phi}(y) \right) \bigg \rvert
	\leq\\
	\sum_{x \in \{0,1\}^n}
	\bigg \lvert \hat{\phi}(x)-\rho(x) \bigg \rvert
	+\left \lvert 1-\sum_{y \in \{0,1\}^n} \hat{\phi}(y) \right \rvert
	=\\
	\sum_{x \in \{0,1\}^n}
	\bigg \lvert \hat{\phi}(x)-\rho(x) \bigg \rvert
	+
	\left \lvert
	\sum_{x \in \{0,1\}^n}
	\rho(x)-\hat{\phi}(x)
	\right \rvert
	\leq
	2 \cdot \sum_{x \in \{0,1\}^n}
	\bigg \lvert \hat{\phi}(x)-\rho(x) \bigg \rvert
	\leq 2 \cdot \epsilon
	\end{multline*}
	Thus, with probability at least $2/3$, we have $d_{\text{TV}}(\rho, \hat{\rho}) \leq \epsilon$.
\end{proof}

\subsection{Main lemma}
\label{section: slackness regret}
Here we prove the following structural lemma. The lemma, as well as its proof are inspired by the main structural lemma of \cite{blais2014dnf} (i.e. Lemma \ref{lemma BHST}). Recall that the slack of a monotone function was given in Definition \ref{definition: monotone slack}.

\ourMainLemma*
\begin{proof}
	We use the following process to obtain $f$:
	\begin{enumerate}
		\item[a)] Set $f^*=\rho$.
		\item[b)] For $h=0$ to $n$:
		\begin{itemize}
			\item If it is the case that:
			\begin{equation}
			\label{equation: process condition}
			\frac{1}{\binom{n}{h}}
			\cdot
			\sum_{x \in \{0,1\}^n \text{ s.t. } ||x||=h} \text{slack}_{f^*}(x)
			<
			R_h \cdot \frac{1}{\sum_{j=h}^n \binom{n}{h}}
			\end{equation}
			Then, for all $x$ in $\{0,1\}^n$, satisfying $||x||=h$ set:
			$
			f^{*}(x) :=
			f^{*}(x)-
			\text{slack}_{f^*}(x)
			$.
		\end{itemize}
		\item[c)] Set $f=f^*$ and output $f$.
	\end{enumerate}
	
	By inspection, $f^*$ remains monotone and positive at every iteration of the process. Therefore, $f$ is also monotone and positive.
	
	Property (1) in the Lemma is true, because at every step of the process, values of $f^*$ only decrease.
	
	To see why
	Property (2) is the case, note that the value $\sum_{x \in \{0,1\}^n}
	\rho(x)-f(x)$ is zero in the beginning of the process, and at a step $h$ it either stays the same or decreases by at most $R_h \cdot \frac{\binom{n}{h}}{\sum_{j=h}^n \binom{n}{j}}$. Therefore we can upper-bound:
	\[
	\sum_{x \in \{0,1\}^n}
	\rho(x)-f(x)
	\leq
	\sum_{h=0}^n
	R_h \cdot \frac{\binom{n}{h}}{\sum_{j=h}^n \binom{n}{j}}
	\leq \zeta
	\]
	
	Now, the only thing left to prove is that property (3) holds. 
	
	From the definition of monotone slack, it follows that modifying the value of a function on points of Hamming weight $j$ does not affect the slack on any point with Hamming weight lower than $j$. Therefore, the value $\frac{1}{\binom{n}{j}}
	\cdot
	\sum_{x \in \{0,1\}^n \text{ s.t. } ||x||=j} \text{slack}_{f^*}(x)
	$ will not change as $f^*$ changes after the $j$th iteration. Therefore, this value will be equal to $\frac{1}{\binom{n}{j}}
	\cdot
	\sum_{x \in \{0,1\}^n \text{ s.t. } ||x||=j} \text{slack}_{f}(x)
	$. Thus, the value of $\frac{1}{\binom{n}{j}}
	\cdot
	\sum_{x \in \{0,1\}^n \text{ s.t. } ||x||=j} \text{slack}_{f}(x)
	$ is either zero or at least $R_h \cdot \frac{1}{\sum_{j=h}^n \binom{n}{h}}$. 
	
	Now, we need the following generalization of Fact \ref{fact: average density increases}:
	\begin{observation}
		\label{observation: generalized basic fact}
		Let $f$ be an arbitrary monotone function $\{0,1\}^n \rightarrow R$. Then, for any $k$ in $[0,n-1]$ it is the case that:
		\[
		\mu_f(k+1)
		\geq
		\mu_f(k)
		+\frac{1}{\binom{n}{k+1}}
		\cdot
		\sum_{x \in \{0,1\}^n \text{ s.t. } ||x||=k+1}
		\text{slack}_f(x)
		\]
	\end{observation}
	\begin{proof}
		For all $x$ with $||x||=k+1$ we have that:
		\[
		f(x)
		=
		\text{slack}_f(x)
		+
		\max_{y \in \{0,1\}^n \text{ s.t. } ||y||=k \text{ and } y \preceq x} f(y)
		\]
		We have that:
		\[
		\max_{y \in \{0,1\}^n \text{ s.t. } ||y||=k \text{ and } y \preceq x} f(y)
		\geq
		\mathbb{E}_{y \sim \{0,1\}^n \text{ conditioned on } ||y||=k \text{ and } y \preceq x}[f(y)]
		\]
		Therefore:
		\[
		f(x)
		\geq
		\text{slack}_f(x)
		+
		\mathbb{E}_{y \sim \{0,1\}^n \text{ conditioned on } ||y||=k \text{ and } y \preceq x}[f(y)]
		\]
		Averaging the both sides, we get:
		\begin{multline}
		\mu_f(k+1)
		\geq
		\frac{1}{\binom{n}{k+1}}
		\cdot
		\sum_{x \in \{0,1\}^n \text{ s.t. } ||x||=k+1}
		\text{slack}_f(x)
		+\\
		\mathbb{E}_{x \sim \{0,1\}^n \text{ conditioned on } ||x||=k+1}
		\mathbb{E}_{y \sim \{0,1\}^n \text{ conditioned on } ||y||=k \text{ and } y\preceq x}[f(y)]=\\
		\frac{1}{\binom{n}{k+1}}
		\cdot
		\sum_{x \in \{0,1\}^n \text{ s.t. } ||x||=k+1}
		\text{slack}_f(x)
		+
		\mathbb{E}_{y \sim \{0,1\}^n \text{ conditioned on } ||y||=k}[f(y)]=\\
		\frac{1}{\binom{n}{k+1}}
		\cdot
		\sum_{x \in \{0,1\}^n \text{ s.t. } ||x||=k+1}
		\text{slack}_f(x)
		+
		\mu_f(k)
		\end{multline}
		Above, the penultimate equality followed from a simple probabilistic fact: if one picks a random $n$-bit string of Hamming weight $k+1$ and then sets to zero a random bit that equals to one, this is equivalent to picking a random $n$-bit string of weight $k$. 
	\end{proof}
	
	Using the Observation \ref{observation: generalized basic fact} repeatedly and recalling that in Definition \ref{definition: monotone slack} we defined $\text{slack}_f(0^n)=f(0^n)$, we get that for all $h$:
	\begin{multline*}
	\mu_{f}(h)
	\geq
	\mu_f(0)+
	\sum_{k=1}^h \frac{1}{\binom{n}{k}}
	\cdot
	\sum_{x \in \{0,1\}^n \text{ s.t. } ||x||=k}
	\text{slack}_f(x)
	=\\
	\sum_{k=0}^h \frac{1}{\binom{n}{k}}
	\cdot
	\sum_{x \in \{0,1\}^n \text{ s.t. } ||x||=k}
	\text{slack}_f(x)
	\geq
	\sum_{k=0}^h
	R_k \cdot \frac{1}{\sum_{j=k}^n \binom{n}{j}}
	\cdot
	\mathbf{1}_{\exists \: x \in \{0,1\}^n: \; ||x||=k \: \land \: \text{slack}_f(x) \neq 0}
	\end{multline*}
	Summing this up over all $h$ and changing the order of summations, we get:
	\begin{multline*}
	1=\sum_{x \in \{0,1\}^n} \rho(x)
	\geq \sum_{x \in \{0,1\}^n} f(x)
	= \sum_{h=0}^n \binom{n}{h} \mu_{f}(h)
	\geq\\
	\sum_{h=0}^n \binom{n}{h} 
	\sum_{k=0}^h
	R_k \cdot \frac{1}{\sum_{j=k}^n \binom{n}{j}}
	\cdot
	\mathbf{1}_{\exists \: x \in \{0,1\}^n: \; ||x||=k \: \land \: \text{slack}_f(x) \neq 0}=\\
	\sum_{k=0}^n \sum_{h=k}^n \binom{n}{h} 
	R_k \cdot \frac{1}{\sum_{j=k}^n \binom{n}{j}}
	\cdot
	\mathbf{1}_{\exists \: x \in \{0,1\}^n: \; ||x||=k \: \land \: \text{slack}_f(x) \neq 0}
	=\\
	\sum_{k=0}^n
	R_k \cdot
	\mathbf{1}_{\exists \: x \in \{0,1\}^n: \; ||x||=k \: \land \: \text{slack}_f(x) \neq 0}
	\end{multline*}
	This finishes the proof of the lemma.
\end{proof}

Now, we prove the following corollary:

\corollaryOfOurMainLemma*

\begin{proof}
	We use Lemma \ref{lemma: slack regret}, setting $\zeta=\epsilon/4$ and
	\[
	R_h=
	\begin{cases}
	\frac{\epsilon^2}{16} &\text{ if $h \leq h_0$} \\
	0 &\text{ otherwise}
	\end{cases}
	\]
	We verify the precondition to Lemma \ref{lemma: slack regret}, by using that $\sum_{x \in \{0,1\}^n} \rho(x)-f(x) \leq \frac{\epsilon}{4}$:
	\begin{multline*}
	\sum_{h=0}^n R_h \cdot \frac{\binom{n}{h}}{\sum_{j=h}^n \binom{n}{j}}
	=
	\sum_{h=0}^{h_0} \frac{\epsilon^2}{16} \cdot \frac{\binom{n}{h}}{\sum_{j=h}^n \binom{n}{j}}
	\leq
	\sum_{h=0}^{h_0} \frac{\epsilon^2}{16} \cdot \frac{\binom{n}{h}}{\sum_{j=h_0}^n \binom{n}{j}}
	\leq\\
	\sum_{h=0}^{h_0} \frac{\epsilon^2}{16} \cdot \frac{\binom{n}{h}}{2^n \cdot \epsilon/4}
	=\frac{\epsilon}{4} \cdot
	\sum_{h=0}^{h_0} \frac{\binom{n}{h}}{2^n}
	\leq \frac{\epsilon}{4} 
	\end{multline*}
	Now, we simply check that properties (1), (2) and (3) of the Lemma directly imply the properties (1), (2) and (3) of the Corollary respectively. This completes the proof.
\end{proof}
To use Lemma \ref{lemma: slack regret}, we need an upper bound on the value of $\frac{\binom{n}{h}}{\sum_{j \geq h}^n \binom{n}{j}}$. The following claim provides such an upper bound:
\begin{claim}
	\label{claim: bound on sum of binomials divided by binomial}
	For all sufficiently large $n$, for all $h$, satisfying $0 \leq h \leq n$, it is the case that:
	\[
	\frac{\binom{n}{h}}{\sum_{j \geq h}^n \binom{n}{j}}
	\leq
	\left(\begin{cases}
	\frac{2}{n^2} &\text{ if $h \leq n/2-\sqrt{n \ln(n)}$}
	\\ \frac{200}{\sqrt{n}} &\text{ if $n/2-\sqrt{n \ln(n)}< h < n/2+\sqrt{n}$} \\ 
	200 \cdot \frac{h-n/2}{n} &\text{ if $h \geq n/2+\sqrt{n}$}\end{cases} \right)
	\]
\end{claim}
\begin{proof}
	See Appendix A, subsection \ref{appendix subsection: proof of claim bound on sum of binomials divided by binomial}
	
\end{proof}

\subsection{Proof of claim \ref{claim: learning main claim}}

\label{section: proof lerning main claim}

For all $x$ in $\{0,1\}^n$, satisfying $9 \sqrt{n} \leq ||x||$, we define the following quantity:
\begin{multline}
\label{equation definition of phi learning}
\phi(x)
\myeq
\frac{1}{2^{\left \lfloor L_{||x||}\right \rfloor}} \cdot \max_{y \text{ s.t. } y \preceq x \text{ and } ||x||-||y|| = {\left \lfloor L_{||x||}\right \rfloor}} \text{  }
\Pr_{z \sim \rho}[ y \preceq z \preceq x ]
=\\
\frac{1}{2^{\left \lfloor L_{||x||}\right \rfloor}} \: \cdot \max_{y \text{ s.t. } y \preceq x \text{ and } ||x||-||y|| = {\left \lfloor L_{||x||}\right \rfloor }} \text{  }
\sum_{z \text{ s.t. } y \preceq z \preceq x} \rho(z)
\end{multline}

Observe that
since for every such $x$ and $y$ there are $2^{\left \lfloor L_{||x||} \right \rfloor}$ values of $z$ satisfying $y \preceq z \preceq x $, and $\rho$ is a monotone probability distribution, it has to be the case that $\phi(x) \leq \rho(x)$ for all $x$ on which $\phi(x)$ is defined.

More interestingly,
we will be claiming that $\phi$ is (in terms of $L_1$ distance) a good approximation to $\rho$, but first we will show that $\hat{\phi}$ is a good approximation to $\phi$, assuming that the values $L_h$ are not too small:
\begin{claim}
	\label{claim: phi hat is close to phi}
	If it is the case that $\frac{1}{2^n}
	\cdot
	\sum_{h=9 \sqrt{n}}^n
	\binom{n}{h}
	\cdot
	\frac{A}{2^{L_{h}}} \leq \frac{1}{2}$, 
	then, with probability at least $7/8$, it is the case that:
	\begin{equation}
	\label{equation big union bound for phi learning}
	\sum_{x \in \{0,1\}^n \text{ s.t. } 9 \sqrt{n} \leq ||x||}
	\bigg \lvert \hat{\phi}(x)-\phi(x) \bigg \rvert
	\leq
	\frac{\epsilon}{4}
	\end{equation}
\end{claim}
\begin{proof}
	See Appendix A, subsection \ref{appendix subsection: phi hat is close to phi}, for the proof, which follows using tail bounds.
\end{proof}

Now, we apply Lemma \ref{lemma: slack regret} to $\rho$, with value $\zeta:=\epsilon/100$. For now, we postpone setting the values of $R_h$, which we will do later in our derivation (of course, we will then check that the required constraint is indeed satisfied by these values).

This gives a positive monotone function $f$ that satisfies the three conditions of Lemma \ref{lemma: slack regret}.
We separate all the values of $x$ in $\{0,1\}^n$ for which $9 \sqrt{n} \leq ||x|| $ into two kinds: \textbf{good} and \textbf{bad}. We say that $x$ is \textbf{bad} if there is some $y$ for which $0 \leq ||x||-||y|| <\left \lfloor  L_{||x||} \right \rfloor$ and $\text{slack}_f(y)$ is non-zero. Otherwise, $x$ if  \textbf{good}. Clearly, for a given Hamming weight value, wither every point with this Hamming weight is good, or every such point is bad. 

We can write:
\begin{equation}
\label{equation: decomposition to good and bad learning}
\sum_{x \in \{0,1\}^n \text{ s.t. } 9 \sqrt{n} \leq ||x||}
\lvert \phi(x)-f(x) \rvert
=
\sum_{\text{good } x}
\lvert \phi(x)-f(x) \rvert
+
\sum_{\text{bad } x}
\lvert \phi(x)-f(x) \rvert
\end{equation}

Now, we bound the two terms above separately. If $x$ is good, then it is the case that for all $y$ satisfying $||x||-\lfloor L_{||x||} \rfloor < ||y|| \leq ||x||$ we have $\text{slack}_{f}(x)=0$, and therefore $f(y)=\max_{y' \in \{0,1\}^n \text{ s. t. } y' \preceq y \text{ and } ||y||-||y'||=1} f(y')$. Using this relation recursively, we obtain that:
\[
f(x)=\max_{y \in \{0,1\}^n \text{ s. t. } y \preceq x \text{ and } ||x||-||y||=\lfloor L_{||x||} \rfloor} f(y)
\] 
Therefore, since $f$ is monotone, we obtain that:
\[
f(x)=
\frac{1}{2^{\lfloor L_{||x||} \rfloor}} \cdot \max_{y \text{ s.t. } y \preceq x \text{ and } ||x||-||y|| = \lfloor L_{||x||} \rfloor} \text{  }
\sum_{z \text{ s.t. } y \preceq z \preceq x} f(z)
\]
By Lemma \ref{lemma: slack regret}, it is the case $\rho(x) \geq f(x)$. This, together with the equation above and Equation \ref{equation definition of phi learning} implies:
\[
\phi(x) \geq f(x)
\] 
But we also know that $\rho(x) \geq \phi(x)$. Therefore:
\begin{equation}
\label{equation bound on good points learning}
\sum_{\text{good } x}
\lvert \phi(x)-f(x) \rvert
\leq
\sum_{\text{good } x}
\lvert \rho(x)-f(x) \rvert
\leq
\frac{\epsilon}{4}
\end{equation}
Where the last inequality follows from Lemma \ref{lemma: slack regret}.

Now, we bound the contribution of bad points. Since $\phi(x) \leq \rho(x)$, $f(x) \leq \rho(x)$ and recalling the definition of a bad point, we get:
\begin{multline}
\label{equation: bounding bad one learning}
\sum_{\text{bad } x}
\lvert \phi(x)-f(x) \rvert
\leq 
\sum_{\text{bad } x}
\max(\phi(x),f(x))
\leq 
\sum_{\text{bad } x}
\rho(x)
\leq\\
\sum_{h_2=9 \sqrt{n} }^n \mu_\rho(h_2) \cdot \binom{n}{h_2} \cdot
\mathbf{1}_{\exists \: x \in \{0,1\}^n: \; \left(h_2-\left \lfloor L_{h_2} \right \rfloor < ||x|| \leq h_2\right) \: \land \: \text{slack}_f(x) \neq 0}
\end{multline}
Since Lemma \ref{lemma: slack regret} gives us a bound on a weighed sum of indicator variables of the form
$ \mathbf{1}_{\exists \: x \in \{0,1\}^n: \; ||x||=h \: \land \: \text{slack}_f(x) \neq 0}$, we would like to upper-bound the expression above by such a weighted sum. To do this, to every Hamming weight value $h$ that has a point $x$ with non-zero $\text{slack}_f(x)$ (we call such Hamming weight value $h$ \textbf{slacky}) we ``charge'' every value $h_2$, for which points of Hamming weight $h_2$ are rendered bad because $h$ is slacky. This will happen only if $h_2 \geq h$ and $h_2-\left \lfloor L_{h_2} \right \rfloor < h$ . But since $\left \lfloor L_{h_2} \right \rfloor$ can only decrease as $h_2$ increases, the latter can happen only if $h_2 - \left \lfloor L_h \right \rfloor < h$. Therefore:

\begin{equation}
\sum_{\text{bad } x}
\lvert \phi(x)-f(x) \rvert
\leq
\sum_{h=0}^n \left(\sum_{h_2=h}^{h+\left \lfloor L_{h} \right \rfloor-1} \mu_\rho(h_2) \cdot \binom{n}{h_2}  \right)
\cdot
\mathbf{1}_{\exists \: x \in \{0,1\}^n: \; ||x||=h \: \land \: \text{slack}_f(x) \neq 0}
\end{equation}

Now, to upper-bound $\mu_{\rho}(h_2)$, we need the following claim:

\begin{claim}
	\label{claim: upper bound on level density}
	For any monotone probability distribution $\rho$ it is the case that for all $h$:
	\[
	\mu_{\rho}(h)
	\leq
	\frac{1}{\sum_{j=h}^n \binom{n}{j}}
	\]
\end{claim}

\begin{proof}
	This follows immediately from Fact \ref{fact: average density increases} and that $\sum_{x \in \{0,1\}^n} \rho(x)=1$.
\end{proof}
Claim \ref{claim: upper bound on level density}, Equation \ref{equation: bounding bad one learning} and Claim \ref{claim: bound on sum of binomials divided by binomial} together imply:
\begin{multline}
\label{equation: bounding bad 2 learning}
\sum_{\text{bad } x}
\lvert \phi(x)-f(x) \rvert
\leq
\sum_{h=0}^n \left(\sum_{h_2=h}^{h+\left \lfloor L_{h} \right \rfloor-1} \frac{\binom{n}{h_2}}{\sum_{j=h_2}^n \binom{n}{j}}  \right)
\mathbf{1}_{\left(\substack{\exists \: x \in \{0,1\}^n:\\ \; ||x||=h \: \land \: \text{slack}_f(x) \neq 0}\right)}
\; \leq\\
\sum_{h=0}^n \left(\sum_{h_2=h}^{h+\left \lfloor L_{h} \right \rfloor-1} \left(\begin{cases} \frac{200}{\sqrt{n}} &\text{ if $h_2 < n/2+\sqrt{n}$} \\ 200 \cdot \frac{h_2-n/2}{n} &\text{ if $h_2 \geq n/2+\sqrt{n}$} \end{cases} \right)  \right)
\mathbf{1}_{\left(\substack{\exists \: x \in \{0,1\}^n:\\ \; ||x||=h \: \land \: \text{slack}_f(x) \neq 0}\right)}
\; \leq \\
\sum_{h=0}^n L_h \cdot \left(\begin{cases} \frac{200}{\sqrt{n}} &\text{ if $h+L_h < n/2+\sqrt{n}$} \\ 200 \cdot \frac{h+L_h-n/2}{n} &\text{ if $h+L_h \geq n/2+\sqrt{n}$}\end{cases} \right)  
\mathbf{1}_{\left(\substack{\exists \: x \in \{0,1\}^n:\\ \; ||x||=h \: \land \: \text{slack}_f(x) \neq 0}\right)}
\end{multline}

Now, we claim that:
\begin{equation}
\label{equation: two cases adding L_h does not increase by more than 10}
\left(\begin{cases} \frac{200}{\sqrt{n}} &\text{ if $h+L_h < n/2+\sqrt{n}$} \\ 200 \cdot \frac{h+L_h-n/2}{n} &\text{ if $h+L_h \geq n/2+\sqrt{n}$}\end{cases} \right)
\leq
10 \cdot
\left(\begin{cases} \frac{200}{\sqrt{n}} &\text{ if $h < n/2+\sqrt{n}$} \\ 200 \cdot \frac{h-n/2}{n} &\text{ if $h\geq n/2+\sqrt{n}$}\end{cases} \right)
\end{equation}
This follows by considering three cases (i) $h+L_h < n/2+\sqrt{n}$, in which case this is equivalent to $\frac{200}{\sqrt{n}} \leq \frac{2000}{\sqrt{n}}$, which is trivially true. (ii) $h \geq n/2 + \sqrt{n}$, in which case since $L_h \leq 9 \sqrt{n}$, we have that $\frac{h+L_h-n/2}{n}\leq 10 \cdot \frac{h-n/2}{n}$ (iii) $h+L_h \geq n/2 +\sqrt{n}$, but $h < n/2 +\sqrt{n}$, in which case since $L_h \leq 9 \sqrt{n}$, we have that $\frac{h+L_h-n/2}{n}\leq \frac{\sqrt{n}+L_h}{n}\leq{10}{\sqrt{n}}$.

Combining Equations \ref{equation: bounding bad 2 learning} and \ref{equation: two cases adding L_h does not increase by more than 10}, we get:
\begin{multline}
\label{equation: bounding bad 3 learning}
\sum_{\text{bad } x}
\lvert \phi(x)-f(x) \rvert
\leq\\
\sum_{h=0}^n L_h \cdot 10 \cdot \left(\begin{cases} \frac{200}{\sqrt{n}} &\text{ if $h < n/2+\sqrt{n}$} \\ 200 \cdot \frac{h-n/2}{n} &\text{ otherwise} \end{cases} \right) \cdot 
\mathbf{1}_{\exists \: x \in \{0,1\}^n: \; ||x||=h \: \land \: \text{slack}_f(x) \neq 0}
\end{multline}
Recall that we postponed setting the values of $R_h$. The equation above motivates us to set:
\[
R_h:=\frac{200}{\epsilon} \cdot L_h \cdot \left(\begin{cases} \frac{200}{\sqrt{n}} &\text{ if $h < n/2+\sqrt{n}$} \\ 200 \cdot \frac{h-n/2}{n} &\text{ otherwise} \end{cases} \right) 
\]
Now, we check the constraint on $R_h$ in Lemma \ref{lemma: slack regret}. Using Claim \ref{claim: bound on sum of binomials divided by binomial} and the premise of Claim \ref{claim: learning main claim}:

\begin{multline*}
\sum_{h=0}^n R_h \cdot \frac{\binom{n}{h}}{\sum_{j \geq h}^n \binom{n}{j}}
\leq
\sum_{h=0}^n R_h \cdot \left(\begin{cases}
\frac{2}{n^2} &\text{ if $h \leq n/2-\sqrt{n \ln(n)}$}
\\ \frac{200}{\sqrt{n}} &\text{ if $n/2-\sqrt{n \ln(n)}< h < n/2+\sqrt{n}$} \\ 
200 \cdot \frac{h-n/2}{n} &\text{ if $h \geq n/2+\sqrt{n}$}\end{cases} \right)
= \\
\frac{200}{\epsilon}
\cdot
\sum_{h=0}^n L_h \cdot \left(\begin{cases}
\frac{400}{n^{2.5}} &\text{ if $h \leq n/2-\sqrt{n \ln(n)}$}
\\ \frac{40000}{n} &\text{ if $n/2-\sqrt{n \ln(n)}< h < n/2+\sqrt{n}$} \\ 
40000 \cdot \left(\frac{h-n/2}{n} \right)^2 &\text{ if $h \geq n/2+\sqrt{n}$}\end{cases} \right)
\leq \frac{\epsilon}{100} =\zeta
\end{multline*}
Therefore, Lemma \ref{lemma: slack regret}, together with Equation \ref{equation: bounding bad 3 learning} implies that:

\begin{equation*}
\sum_{\text{bad } x}
\lvert \phi(x)-f(x) \rvert
\leq 
\sum_{h=0}^n \frac{\epsilon}{20} \cdot R_h  \cdot 
\mathbf{1}_{\exists \: x \in \{0,1\}^n: \; ||x||=h \: \land \: \text{slack}_f(x) \neq 0}
\leq
\frac{\epsilon}{20}
\end{equation*}
Now, using triangle inequality and then combining the inequality above with Equations \ref{equation: chernoff plus hoeffding learning}, \ref{equation: decomposition to good and bad learning} and \ref{equation bound on good points learning} we get:

\begin{multline}
\label{bound on everything except super low learning}
\sum_{x \in \{0,1\}^n \text{ s.t. } 9 \sqrt{n} \leq ||x||}
\bigg \lvert \hat{\phi}(x)-\rho(x) \bigg \rvert
\leq\\
\sum_{x \in \{0,1\}^n \text{ s.t. } 9 \sqrt{n} \leq ||x||}
\bigg \lvert \hat{\phi}(x)-\phi(x) \bigg \rvert
+
\sum_{x \in \{0,1\}^n \text{ s.t. } 9 \sqrt{n} \leq ||x||}
\bigg \lvert \phi(x)-\rho(x) \bigg \rvert \leq \\
\frac{\epsilon}{4}+\frac{\epsilon}{100}+\frac{\epsilon}{20} \leq \frac{\epsilon}{2}
\end{multline}
\section{Estimating the distance to uniform}
\begin{figure}
	\caption{}
	\label{algorithm: estimating the distance to uniform efficiently}
	\hfill \newline
			\textbf{Algorithm for the estimation of distance to uniform efficiently} (given sample access from a distribution $\rho$, which is monotone over $\{0,1\}^n$.)
			
			\dotfill
			
			\begin{enumerate}
				\item Pick set $h_0$ to be an integer for which it is the case that:
				\begin{equation}
				\label{equation: h_0 is good}
				\frac{\epsilon}{4}
				\leq
				\Pr_{x \sim \{0,1\}^n}
				[||x|| \geq h_0]
				\leq
				\frac{\epsilon}{2}
				\end{equation}
				Do this by going through every integer candidate $h_{\text{candidate}}$ in the interval and computing the fraction of points $x$ in $\{0,1\}$ for which $||x|| \geq h_{\text{candidate}}$. Finally, pick $h_0$ to be one of $h_{\text{candidate}}$ for which the relation above holds.
				\item Set $N_1:=\frac{32 \ln 2}{\epsilon^2}$. Draw $N_1$ samples from the probability distribution $\rho$ and denote the multiset of these samples as $S_1$.
				\item Set:
				\[
				\hat{d}_1 :=\frac{1}{2} \cdot \frac{
					\bigg \lvert \bigg \{z \in S_1: ||z|| \geq h_0 \bigg\} \bigg \rvert}{N_1}
				\]
				\item Set $L:= \left \lfloor \frac{\sqrt{n} \epsilon^4}{512} \right \rfloor$.
				\item Set 
				\[N_2:=\frac{2^n}{2^L} \cdot \frac{192}{\epsilon^2}
				\cdot
				\bigg(
				n \ln 2+L \ln n +4 \ln 2
				\bigg)
				\]
				Draw $N_2$ samples from the probability distribution $\rho$ and denote the multiset of these samples as $S_2$.
				\item For all $x$, satisfying $L \leq ||x|| < h_0$, compute:
				\[
				\hat{\phi}(x)
				:=
				\frac{1}{2^L} \cdot \frac{\max_{y \text{ s.t. } y \preceq x \text{ and } ||x||-||y|| = L}
					\bigg \lvert \bigg \{z \in S_2: y \preceq z \preceq x \bigg\} \bigg \rvert}{N_2}
				\] 
				Do this by first making a look-up table, which given arbitrary $z \in \{0,1\}^n$ returns the number of times $z$ was encountered in $S_2$. Then, use this look-up table to compute the necessary values of $\lvert \{z \in S_2: y \preceq z \preceq x \} \rvert$  by querying all these values of $z$ in the lookup table and summing the results up. 
				\item Compute the following:
				\[
				\hat{d}_2
				:=
				\frac{1}{2} \cdot
				\sum_{x \text{ s.t. } L \leq ||x|| < h_0}
				\bigg \lvert
				\hat{\phi}(x)-\frac{1}{2^n}
				\bigg \rvert
				\]
				\item Output $\hat{d}_1+\hat{d}_2$.
			\end{enumerate}
		
		\dotfill
\end{figure}
In this section we prove our upper-bound on the sample complexity of estimating the distance from uniform of an unknown monotone probability distribution over the Boolean cube. We restate the theorem:
\theoremDistanceToUniformUpperBound*
\begin{proof}
	We present the algorithm in Figure \ref{algorithm: estimating the distance to uniform efficiently}. The number of samples drawn from $\rho$ is $N_1+N_2=\frac{2^{n}}{2^{\Theta_{\epsilon}(\sqrt{n})}}$. The run-time, in turn, is dominated\footnote{Step 1 requires only $2^n\poly(n)$ time, which is less than what step (6) requires. By inspection, other steps require even less run-time. Incidentally, the task in step 1 can be done much faster by randomized sampling, but since this is not the run-time bottleneck, we use this direct approach for the sake of simplicity.} by computing the values of $\hat{\phi}$ in step (6), in which the construction of the lookup table takes $O(n \cdot 2^n)$ time and the time spent computing each $\hat{\phi}(x)$ can be upper bounded by the product of: (i) the number of pairs $(y,z)$ that simultaneously satisfy $y \preceq z \preceq x$ and $||x||-||y|| =L$, which can be upper-bounded by $O(n^L \cdot 2^L)$ and (ii) the time it takes to look up a given $z$ in the lookup table, which can be upper-bounded by $O(n)$. Overall, this gives us a run-time upper bound of $O(2^{n+O_{\epsilon}(\sqrt{n} \log n
		)})$.
	
	Now, the only thing left to prove is correctness. First of all, it is not a priori clear that there exists a value of $h_0$ satisfying Equation \ref{equation: h_0 is good} (on Figure \ref{algorithm: estimating the distance to uniform efficiently}). This is true for the following reason: imagine changing $h_{\text{candidate}}$ from $n$ to $0$ by decrementing it in steps of one. Then $\Pr_{x \in \{0,1\}^n}[||x|| \geq h_{\text{candidate}}]$ will increase from $\frac{1}{2^n}$ to $1$ and by Fact \ref{fact: boolean cube anti-concentration} it will not increase by more than $\frac{2}{\sqrt{n}}$ at any given step. For sufficiently large $n$ we have $\frac{2}{\sqrt{n}}<\frac{\epsilon}{4}$. Then it is impossible to skip over the interval between $\frac{\epsilon}{4}$ and $\frac{\epsilon}{2}$ in just one step of length at most $\frac{2}{\sqrt{n}}$, and therefore Equation \ref{equation: h_0 is good} (on Figure \ref{algorithm: estimating the distance to uniform efficiently}) will be the case for some value of $h_{\text{candidate}}$.

	We decompose the total variation distance between $\rho$ and the uniform distribution into three terms:
	\begin{multline}
	\label{equation: grand scheme for distance to uniform}
	\frac{1}{2} \cdot \sum_{x \in \{0,1\}^n}
	\bigg \lvert \rho(x)-\frac{1}{2^n} \bigg \rvert
	=\\
	\frac{1}{2} \cdot \sum_{\substack{x \in \{0,1\}^n \text{ s.t. }\\ ||x|| \geq h_0}}
	\bigg \lvert \rho(x)-\frac{1}{2^n} \bigg \rvert
	+
	\frac{1}{2} \cdot \sum_{\substack{x \in \{0,1\}^n \text{ s.t. }\\ L \leq ||x|| < h_0}}
	\bigg \lvert \rho(x)-\frac{1}{2^n} \bigg \rvert
	+
	\frac{1}{2} \cdot \sum_{\substack{x \in \{0,1\}^n \text{ s.t. }\\ ||x|| < L}}
	\bigg \lvert \rho(x)-\frac{1}{2^n} \bigg \rvert
	\end{multline}
	We argue that the first term is well approximated by $\hat{d_1}$, the second term is well approximated by $\hat{d_2}$, and the third term is negligible. As the reader will see, out of these three terms, the middle term is the least trivial to prove guarantees for.  
	
	We will first handle the first term: From the triangle inequality, Hoeffding's bound and Equation \ref{equation: h_0 is good} (on Figure \ref{algorithm: estimating the distance to uniform efficiently}) it follows immediately that with probability at least $7/8$ it is the case that:
	\begin{multline}
	\label{equation: bound on d 1}
	\bigg \lvert
	\hat{d}_1
	-
	\frac{1}{2} \cdot
	\sum_{x \in \{0,1\}^n \text{ s.t. } ||x|| \geq h_0}
	\bigg \lvert \rho(x)-\frac{1}{2^n} \bigg \rvert
	\bigg \rvert
	\leq \\
	\bigg \lvert
	\hat{d}_1
	-
	\frac{1}{2} \cdot
	\sum_{x \in \{0,1\}^n \text{ s.t. } ||x|| \geq h_0}
	\rho(x) 
	\bigg \rvert
	+
	\frac{1}{2} \cdot
	\sum_{x \in \{0,1\}^n \text{ s.t. } ||x|| \geq h_0}
	\frac{1}{2^n}
	\leq
	\frac{\epsilon}{8}
	+
	\frac{\epsilon}{4}
	=
	\frac{3\epsilon}{8}
	\end{multline}
	
	Now, we use the two following facts: (i) Since $\sum_x \rho(x)=1$ and $\rho$ is monotone, for every $x$ with $||x|| \leq L$ it should be the case that $\rho(x) \leq \frac{1}{2^{n-L}}$. (ii) The number of different values of $x$ in $\{0,1\}^n$ for which $||x|| \leq L$ can be upper bounded by $n^{L}$. We get for sufficiently large $n$:
	\begin{multline}
	\label{equation: bound on TV contribution of the very low points}
	\frac{1}{2} \cdot
	\sum_{x \in \{0,1\}^n \text{ s.t. } ||x|| < L}
	\bigg \lvert \rho(x)-\frac{1}{2^n} \bigg \rvert
	\leq
	\frac{1}{2} \cdot
	\sum_{x \in \{0,1\}^n \text{ s.t. } ||x|| < L}
	\left( \frac{1}{2^n}+ \rho(x) \right)
	\leq\\
	\frac{1}{2} \cdot
	n^L \cdot \left(\frac{1}{2^n}+\frac{1}{2^{n-L}} \right)
	=
	o(1)
	\leq \frac{\epsilon}{8}
	\end{multline}
	
	The rest of this section will be dedicated to proving the following claim:
	
	\begin{claim}
		\label{claim: main technical claim}
		With probability at least $7/8$ it is the case that:
		\begin{equation}
		\label{equation: bound on d 2}
		\bigg \lvert
		\hat{d}_2
		-
		\frac{1}{2} \cdot
		\sum_{x \in \{0,1\}^n \text{ s.t. } L \leq ||x|| < h_0}
		\bigg \lvert \rho(x)-\frac{1}{2^n} \bigg \rvert
		\bigg \rvert
		\leq \frac{\epsilon}{2}
		\end{equation}
	\end{claim}
	Once this is proven, it follows by a union bound that with probability at least $3/4$ both Equations \ref{equation: bound on d 1} and \ref{equation: bound on d 2} will be the case. This, together with Equation \ref{equation: bound on TV contribution of the very low points}, when substituted into Equation \ref{equation: grand scheme for distance to uniform} will imply that:
	\[
	\bigg \lvert
	\sum_{x \in \{0,1\}^n}
	\bigg \lvert \rho(x)- \frac{1}{2^n} \bigg \rvert-(\hat{d}_1+\hat{d}_2)
	\bigg \rvert
	\leq \epsilon
	\]
	This will imply the correctness of our algorithm.
\end{proof}

\subsection{Proof of Claim \ref{claim: main technical claim}}
For all $x$ in $\{0,1\}^n$, satisfying $L \leq ||x|| < h_0$, we define the following quantity:
\begin{multline}
\label{equation definition of phi}
\phi(x)
\myeq
\frac{1}{2^L} \cdot \max_{y \text{ s.t. } y \preceq x \text{ and } ||x||-||y|| = L} \text{  }
\Pr_{z \sim \rho}[ y \preceq z \preceq x ]
=\\
\frac{1}{2^L} \cdot \max_{y \text{ s.t. } y \preceq x \text{ and } ||x||-||y|| = L} \text{  }
\sum_{z \text{ s.t. } y \preceq z \preceq x} \rho(z)
\end{multline}

Observe that
since for every such $x$ and $y$ there are $2^L$ values of $z$ satisfying $y \preceq z \preceq x $, and $\rho$ is a monotone probability distribution, it has to be the case that $\phi(x) \leq \rho(x)$ for all $x$ on which $\phi(x)$ is defined.

We will be claiming that $\phi(x)$ is (in terms of $L_1$ distance) a good approximation to $\rho(x)$, but first we will show that $\hat{\phi}(x)$ is a good approximation to $\phi(x)$:
\begin{claim}
	With probability at least $7/8$, it is the case that:
	\begin{equation}
	\label{equation big union bound for phi}
	\sum_{x \in \{0,1\}^n \text{ s.t. } L \leq ||x|| < h_0}
	\bigg \lvert \hat{\phi}(x)-\phi(x) \bigg \rvert
	\leq
	\frac{\epsilon}{4}
	\end{equation}
\end{claim}
\begin{proof}
	We claim that for any pair $(x,y)$, such that $\phi$ is defined on $x$ and $||x||-||y|| = L$, with probability at least $1-\frac{1}{8\cdot 2^n \cdot n^L}$ the following holds:
	\begin{equation}
	\label{equation: chernoff plus hoeffding}
	\frac{1}{2^L}
	\bigg \lvert
	\Pr_{z \sim \rho}[ y \preceq z \preceq x ]
	-
	\frac{
		\lvert  \{z \in S: y \preceq z \preceq x \}  \rvert}{N_2}
	\bigg \rvert
	\leq 
	\frac{\epsilon}{8} \cdot \max \left(\frac{1}{2^n}
	,\frac{1}{2^L} \Pr_{z \sim \rho}[ y \preceq z \preceq x ] \right)
	\end{equation}
	We use Chernoff's bound to prove this as follows. Denote by $q$ the value $\Pr_{z \sim \rho}[ y \preceq z \preceq x ]$. If $q \geq \frac{2^L}{2^n}$ then by Chernoff's bound we have:
	\begin{multline*}
	\Pr \left[
	\bigg \lvert
	\lvert  \{z \in S: y \preceq z \preceq x \}  \rvert
	-qN_2
	\bigg \rvert \geq \frac{\epsilon}{8} qN_2
	\right]
	\leq
	2\exp\left(-\frac{1}{3} \left(\frac{\epsilon}{8} \right)^2 qN_2 \right)
	\leq\\
	2\exp\left(-\frac{1}{3} \left(\frac{\epsilon}{8} \right)^2 \frac{2^L}{2^n} \cdot N_2 \right)
	=
	\frac{1}{8\cdot 2^n \cdot n^L}
	\end{multline*}
	Otherwise, if we have $q < \frac{2^L}{2^n}$, then by Chernoff's bound:
	\begin{multline*}
	\Pr \left[
	\bigg \lvert
	\lvert  \{z \in S: y \preceq z \preceq x \}  \rvert
	-qN_2
	\bigg \rvert \geq \frac{\epsilon}{8} \cdot \frac{2^L}{2^n} \cdot N_2
	\right]
	\leq
	2\exp\left(-\frac{1}{3} \left(\frac{\epsilon}{8} \cdot \frac{2^L}{2^n} \cdot \frac{1}{q} \right)^2 q N_2 \right)
	\leq\\
	2\exp\left(-\frac{1}{3} \left(\frac{\epsilon}{8} \right)^2 \frac{2^L}{2^n} \cdot N_2 \right)
	=
	\frac{1}{8\cdot 2^n \cdot n^L}
	\end{multline*}
	
	Now, by taking a union bound, it follows that with probability $7/8$ for all such pairs $(x,y)$ Equation \ref{equation: chernoff plus hoeffding} will be the case. For all $x$ on which $\phi$ is defined it then will be the case that:
	\[
	\bigg \lvert
	\hat{\phi}(x)-\phi(x)
	\bigg \rvert
	\leq 
	\frac{\epsilon}{8} \cdot \max \left(\frac{1}{2^n}
	,\phi(x) \right)
	\]
	Summing this for all $x$ in the domain of $\phi$ we get:
	\begin{multline*}
	\sum_{x \in \{0,1\}^n \text{ s.t. } L \leq ||x|| < h_0}
	\bigg \lvert \hat{\phi}(x)-\phi(x) \bigg \rvert
	\leq
	\frac{\epsilon}{8}
	\cdot
	\left(
	2^n \cdot \frac{1}{2^n}
	+
	\sum_{x \in \{0,1\}^n \text{ s.t. } L \leq ||x|| < h_0} \phi(x)
	\right)
	\leq\\
	\frac{\epsilon}{8}
	\cdot
	\left(
	1
	+
	\sum_{x \in \{0,1\}^n \text{ s.t. } L \leq ||x|| < h_0} \rho(x)
	\right)
	\leq
	\frac{\epsilon}{4}
	\end{multline*}
\end{proof}

Now, we apply Corollary \ref{corollary: slack regularity old lemma} to $\rho$. This gives a positive monotone function $f$ that satisfies the three conditions of Corollary \ref{corollary: slack regularity old lemma}.
We separate all the values of $x$ in $\{0,1\}^n$ for which $L \leq ||x|| < h_0 $ into two kinds: \textbf{good} and \textbf{bad}. Recall that by Corollary \ref{corollary: slack regularity old lemma} an element $x$ of $\{0,1\}^n$ for which $L \leq ||x|| < h_0 $ can have $\text{slack}_{f}(x)\neq 0$ only if $||x||=k_i$ for some $i$ between $1$ and $t$. We say that $x$ is \textbf{bad} if there is some $k_i$ for which $0 \leq ||x||-k_i \leq L$. Otherwise, $x$ if  \textbf{good}.

We can write:
\begin{equation}
\label{equation: decomposition to good and bad}
\sum_{x \in \{0,1\}^n \text{ s.t. } L \leq ||x|| < h_0}
\lvert \phi(x)-f(x) \rvert
=
\sum_{\text{good } x}
\lvert \phi(x)-f(x) \rvert
+
\sum_{\text{bad } x}
\lvert \phi(x)-f(x) \rvert
\end{equation}

Now, we bound the two terms above separately. If $x$ is good, then it is the case that for all $z$ satisfying $||x||-L \leq ||z||\leq ||x||$ we have $\text{slack}_{f}(x)=0$, and therefore $f(z)=\max_{z' \in \{0,1\}^n \text{ s. t. } z' \preceq z \text{ and } ||z||-||z'||=1} f(z')$. Using this relation recursively, we obtain that:
\[
f(x)=\max_{z \in \{0,1\}^n \text{ s. t. } z \preceq x \text{ and } ||x||-||z||=L} f(z)
\] 
Therefore, since $f$ is monotone, we obtain that:
\[
f(x)=
\frac{1}{2^L} \cdot \max_{y \text{ s.t. } y \preceq x \text{ and } ||x||-||y|| = L} \text{  }
\sum_{z \text{ s.t. } y \preceq z \preceq x} f(z)
\]
By Corollary \ref{corollary: slack regularity old lemma}, it is the case $\rho(x) \geq f(x)$. This, together with the equation above and Equation \ref{equation definition of phi} implies:
\[
\phi(x) \geq f(x)
\] 
But we also know that $\rho(x) \geq \phi(x)$. Therefore:
\begin{equation}
\label{equation bound on good points}
\sum_{\text{good } x}
\lvert \phi(x)-f(x) \rvert
\leq
\sum_{\text{good } x}
\lvert \rho(x)-f(x) \rvert
\leq
\frac{\epsilon}{4}
\end{equation}
Where the last inequality follows from Corollary \ref{corollary: slack regularity old lemma}.

Now, we bound the contribution of bad points. 

\begin{multline*}
\sum_{\text{bad } x}
\lvert \phi(x)-f(x) \rvert
\leq 
\sum_{\text{bad } x}
\max(\phi(x),f(x))
\leq 
\sum_{\text{bad } x}
\rho(x)
=\\
\sum_{k \in [L, h_0] \text{ s.t. for some $k_i$: $\lvert k-k_i\rvert \leq L$} } \mu_\rho(k) \cdot \binom{n}{k}
\end{multline*}

Now, by Claim \ref{claim: upper bound on level density} we have $\mu_\rho(k) \leq \frac{1}{2^n} \cdot \frac{4}{\epsilon}$ and by Fact \ref{fact: boolean cube anti-concentration} we have that $\binom{n}{k} \leq \frac{2}{\sqrt{n}} \cdot 2^n$. Combining these two facts with the inequality above we get:

\[\sum_{\text{bad } x}
\lvert \phi(x)-f(x) \rvert
\leq
\left(
\frac{1}{2^n} \cdot \frac{4}{\epsilon}
\right) \cdot
\left(
\frac{2}{\sqrt{n}} \cdot 2^n
\right)
\cdot
\left(
L \cdot t
\right)
=
\frac{4}{\epsilon} \cdot
\frac{2}{\sqrt{n}}
\cdot
L \cdot t 
\]
Substituting the value of $L$ and the upper bound on $t$ from Corollary \ref{corollary: slack regularity old lemma} we get:
\[
\sum_{\text{bad } x}
\lvert \phi(x)-f(x) \rvert
\leq
\frac{4}{\epsilon} \cdot
\frac{2}{\sqrt{n}}
\cdot
\frac{\epsilon^4 \sqrt{n}}{512}
\cdot \frac{16}{\epsilon^2}
=
\frac{\epsilon}{4}
\]
Combining this with Equations \ref{equation: decomposition to good and bad} and \ref{equation bound on good points} we get:
\begin{equation}
\label{equation: phi and f are close to each other}
\sum_{x \in \{0,1\}^n \text{ s.t. } L \leq ||x|| < h_0}
\lvert \phi(x)-f(x) \rvert
\leq
\frac{\epsilon}{2}
\end{equation}

Overall, we have:
\begin{multline*}
\bigg \lvert
\sum_{x \in \{0,1\}^n \text{ s.t. } L \leq ||x|| < h_0}
\lvert \rho(x)-1/2^n  \rvert
-
\sum_{x \in \{0,1\}^n \text{ s.t. } L \leq ||x|| < h_0}
\lvert \hat{\phi}(x)-1/2^n  \rvert
\bigg \rvert
\leq\\
\sum_{x \in \{0,1\}^n \text{ s.t. } L \leq ||x|| < h_0}
\lvert \hat{\phi}(x)- \rho(x)  \rvert
\leq
\sum_{x \in \{0,1\}^n \text{ s.t. } L \leq ||x|| < h_0}
\lvert \hat{\phi}(x) - \phi(x)  \rvert
+\\
\sum_{x \in \{0,1\}^n \text{ s.t. } L \leq ||x|| < h_0}
\lvert \phi(x)-f(x)  \rvert
+
\sum_{x \in \{0,1\}^n \text{ s.t. } L \leq ||x|| < h_0}
\lvert f(x) - \rho(x)  \rvert
\end{multline*}
This three terms can be bound using respectively Equation \ref{equation big union bound for phi}, Corollary \ref{corollary: slack regularity old lemma} and Equation \ref{equation: phi and f are close to each other}. This gives us:
\begin{multline*}
\bigg \lvert
2 \cdot 
\hat{d}_2
-
\sum_{x \in \{0,1\}^n \text{ s.t. } L \leq ||x|| < h_0}
\bigg \lvert \rho(x)-\frac{1}{2^n} \bigg \rvert
\bigg \rvert
=\\
\bigg \lvert
\sum_{x \in \{0,1\}^n \text{ s.t. } L \leq ||x|| < h_0}
\lvert \rho(x)-1/2^n  \rvert
-
\sum_{x \in \{0,1\}^n \text{ s.t. } L \leq ||x|| < h_0}
\lvert \hat{\phi}(x)-1/2^n  \rvert
\bigg \rvert
\leq
\epsilon
\end{multline*}

Therefore, with probability at least $7/8$ Equation \ref{equation: bound on d 2} holds, which proves Claim \ref{claim: main technical claim} and completes the proof of correctness.
\section{Estimating the support size}
In this section we prove our upper-bound on the sample complexity of estimating the support size of an unknown monotone probability distribution over the Boolean cube. Recall that a probability distribution $\rho$ is \textbf{well-behaved} if for every $x$ either $\rho(x)=0$ or $\rho(x) \geq 1/2^n$.
We restate the theorem:
\theoremSupportUpperBound*
\begin{proof}
	The algorithm we use is listed in Figure \ref{algorithm: estimating the support size}.
	\begin{figure}
		\caption{}
		\label{algorithm: estimating the support size}
		\hfill \newline
				\textbf{Algorithm for the estimation of support size.} (given sample access from the distribution.)
				
				\dotfill
				
				\begin{enumerate}
					\item Set \[M_1= \frac{2^n}{2^{\frac{ \epsilon^2}{64}\sqrt{n}}} \left(\ln{\frac{32}{\epsilon}}+1 \right)\]
					\item Take $M_1$ samples from the probability distributions. Call the set of these samples $S_1$.
					\item Set \[M_2=\frac{32 \ln 2}{\epsilon^2}\]
					\item Pick $M_2$ elements of $\{0,1\}^n$ uniformly at random. Call these samples $S_2$.
					\item We say that a point $y$ is \textbf{covered} if in $S_1$ there exists at least one $z$, so that $z \preceq y$. One can check if a point $y$ is covered by going through all the $M_1$ elements in $S_1$. Using this checking procedure, compute the fraction $\hat{\eta}$ of the elements in $S_2$ that are covered.
					\item Output $\hat{\eta}$.
				\end{enumerate}

\dotfill
	
\end{figure}
	
	Clearly, the sample complexity is: 
	\[O\left(\frac{2^n}{2^{\frac{ \epsilon^2}{64}\sqrt{n}}} \left(\ln{\frac{32}{\epsilon}}+1 \right)\right)=\frac{2^n}{2^{\Theta_{\epsilon}(\sqrt{n})}}\]
	In turn, the run-time is:
	\[O\left(\frac{2^n}{2^{\frac{ \epsilon^2}{64}\sqrt{n}}} \left(\ln{\frac{32}{\epsilon}}+1 \right)\cdot \frac{32 \ln 2}{\epsilon^2}\right)=\frac{2^n}{2^{\Theta_{\epsilon}(\sqrt{n})}}\]
	
	Now, all is left to prove is correctness. 
	
	Let $\eta$ denote the fraction of elements in $\{0,1\}^n$ that are covered by our samples in $S_1$. Then, a random element of $\{0,1\}^n$ is covered with probability $\eta$. Therefore, by the Hoeffding bound it follows that:
	\begin{equation}
	\label{equation algorithm for support size 1}
	\Pr_{S_2}\left[|\hat{\eta}-\eta|>\frac{\epsilon}{4}\right]
	\leq
	2 \exp\left(-2\left(\frac{\epsilon}{4}\right)^2 M_2\right) = \frac{1}{
		8}
	\end{equation}
	The last equality follows by substituting the value of $M_2$.
	
	Since $\rho$ is monotone, it has to be the case that every point that is covered is in the support of $\rho$. Hence, the support size of $\rho$ is at least $\eta \cdot 2^n$. 
	
	Now, all we need to show is that  $\eta \cdot 2^n$ is not likely to be much smaller than the support size of $\rho$. We call a point $x$ in the support of $\rho$ \textbf{good} if there are at least $2^{\frac{ \epsilon^2}{64} \sqrt{n}}$ points $y$ each of which satisfying: (i) $y$ belongs to the support of $\rho$. (ii) $x \preceq y$. If a point in the support of $\rho$ is not good, then it is \textbf{bad}. We will show that the bad points are few, while a lot of the good points are likely to be covered.
	
	Let $f_{\text{support}}$ be defined as follows:
	\[
	f_{\text{support}} \myeq \begin{cases} 1 \text{ if } \rho(x) \neq 0 \\ 0 \text{ otherwise}\end{cases}
	\]
	In other words, $f_{\text{support}}$ is the indicator function of the support of $\rho$. Since $\rho$ is a monotone probability distribution, $f_{\text{support}}$ is a monotone function. Therefore, applying Lemma \ref{lemma BHST}, there exists a function $g=g_1 \lor ... \lor g_t$ that $\epsilon/4$-approximates $f_{\text{support}}$, where $t \leq 8/\epsilon$ and each $g_i$ is a monotone DNF with terms of width exactly $k_i$. Additionally, $g(x) \leq f(x)$ for all $x$ in $\{0,1\}^n$
	
	\begin{claim}
		For all $i$, $g_i$ contains at most $\frac{\epsilon^2}{32} \cdot 2^n$ bad points. 
	\end{claim}
	\begin{proof}
		Recall that $g_i$ is $k_i$-regular, and therefore every point $x$ on which $g_i(x)=1$ needs to have Hamming weight $\lvert \lvert x \rvert \rvert \geq k_i$.
		
		\begin{claim}
			If $x$ satisfies  $g_i(x)=1$ and $\lvert \lvert x \rvert \rvert \geq k_i+\frac{ \epsilon^2}{64} \sqrt{n}$, then $x$ has to be good
		\end{claim}
		\begin{proof}
			Since $g_i$ is a DNF and $g_i(x)=1$ then $x$ satisfies at least one of the terms of $g_i$. If there are more than one, arbitrarily pick one of them. Let this term be
			\[
			t(y)=\bigwedge_{j \in H_1} y_j
			\] 
			The width of this AND has to be $k_i$, therefore $|H|=k_i$. Since $\lvert \lvert x \rvert \rvert \geq k_i+\frac{\epsilon^2}{64} \sqrt{n}$, there must be $\frac{ \epsilon^2}{64} \sqrt{n}$ values of $j$ for which $x_j =1$ but $j$ is not in $H_1$. Denote the set of these values of $j$ as $H_2$. 
			
			Now, consider an element $y \in\{0,1\}^n$ satisfying the criteria:
			\begin{itemize}
				\item For all $j$ in $H_1$, $y_j=1$.
				\item For all $j$ neither in $H_1$ nor in $H_2$, $y_j=0$.
			\end{itemize}
			Clearly, $t(y)=1$, which implies $g_i(y)=1$, $g(y)=1$, and $f_{\text{support}}(y)=1$. Also, for all $j$, we have $y_j \leq x_j$, and therefore $y \preceq x$. Finally, for all $j$ in $H_2$ the value of $y_j$ can be set arbitrarily to zero or one, and therefore there are $2^{|H_2|}$ such points, which is at least $2^{\frac{ \epsilon^2}{64} \sqrt{n}}$. Therefore, $x$ is a good point.    
			
		\end{proof}  
		
		Thus, we can upper-bound the number of bad points $x$ on which $g_i(x)=1$ by:
		\begin{multline*}
		\bigg \lvert \bigg\{x \in \{0,1\}^n: g(x)=1 \text{ and }k_i+\frac{\epsilon^2}{64} \sqrt{n}  > \sum_j x_j \geq k_i \bigg\} \bigg \rvert
		\leq\\
		\bigg \lvert \bigg\{x \in \{0,1\}^n: k_i+\frac{\epsilon^2}{64} \sqrt{n}  > \sum_j x_j \geq k_i \bigg\} \bigg \rvert
		=\sum_{j=k_i}^{k_i+\frac{\epsilon^2}{64} \sqrt{n}-1} \binom{n}{j} \leq
		\frac{ \epsilon^2}{64} \sqrt{n} \binom{n}{n/2}
		\end{multline*}
		
		By Fact \ref{fact: boolean cube anti-concentration}, for sufficiently large $n$, it is the case that $\binom{n}{n/2} \leq 2 \cdot \frac{2^n}{\sqrt{ n}}$. This implies that the expression above is upper-bounded by $\frac{\epsilon^2}{32} \cdot 2^n$, which completes the proof of this claim.
	\end{proof}
	
	Since $g=g_1 \lor ... \lor g_t(x)$, our claim implies the following:
	
	\begin{multline*}
	\bigg \lvert \bigg\{ x: g(x)=1 \text{ and $x$ is bad}
	\bigg\} \bigg \rvert
	\leq
	\sum_{i=1}^t \bigg \lvert \bigg\{ x: g_i(x)=1 \text{ and $x$ is bad}
	\bigg\} \bigg \rvert
	\leq\\ t \cdot \frac{\epsilon^2}{32} \cdot 2^n
	\leq \frac{8}{\epsilon} \cdot \frac{\epsilon^2}{32} \cdot 2^n=
	\frac{\epsilon}{4} \cdot 2^n
	\end{multline*}
	In addition, there could be at most $\frac{\epsilon}{4} \cdot 2^n$ bad points among the points on which $f_\text{support}$ and $g$ disagree. Thus, in total, there are at most $\frac{\epsilon}{2} \cdot 2^n$ bad points. 
	
	Finally, we need to argue that it is likely that many of the good points get covered:
	\begin{claim}
		Suppose there are $G$ good points. Then, with probability at least $7/8$ it will be the case that at least $1-\epsilon/4$ fraction of these good points are covered.
	\end{claim}
	\begin{proof}
		For every good point $x$ there exist least $2^{\frac{\epsilon^2}{64} \sqrt{n}}$ values of $y$ for which i) $x \preceq y$ and ii) $y$ is in the support of $\rho$. Since $x \preceq y$, if $y$ is ever picked from the distribution, then $x$ will be covered. Since $y$ is in the support of $\rho$, and $\rho$ is well-behaved, we have $\rho(y) \geq \frac{1}{2^n}$. Together, these imply that the probability that a random sample from $\rho$ covers $x$ is at least $\frac{2^{\frac{ \epsilon^2}{64} \sqrt{n}}}{2^n}$. Hence, the probability that any of the $M_1$ i.i.d. samples taken from $\rho$ does not cover $x$ is at most:
		\[
		\left(1- \frac{2^{\frac{\epsilon^2}{64} \sqrt{n}}}{2^n}\right)^{M_1}
		=\left(1- \frac{2^{\frac{\epsilon^2}{64} \sqrt{n}}}{2^n}\right)^{\frac{2^n}{2^{\frac{ \epsilon^2}{64}\sqrt{n}}} \left(\ln{\frac{32}{\epsilon}}+1 \right)}
		\leq
		\frac{1}{e^{\ln \frac{32}{\epsilon}}}=\frac{\epsilon}{
			32}
		\] 
		
		Let $C$ denote a random variable, whose value equals to the number of the good points (out of total $G$) covered after taking $M_1$ i.i.d. samples from $\rho$. 
		
		The value of $C$ has to satisfy these two constraints: (i) It has to be between $0$ and $G$ (ii) By linearity of expectation, $E[C]\geq (1-\frac{\epsilon}{32})G$. Thus, to finish the proof of the Lemma, it is sufficient to show the following claim:
		\begin{claim}
			If, for some fixed $G$, a random variable $C$ is supported on $[0,G]$ and $E[C] \geq (1-\frac{\epsilon}{32})G$, then $\Pr[C\geq( 1- \epsilon/4)G] \geq 7/8$.
		\end{claim}
		\begin{proof}
			This is immediate from Markov's inequality for the random variable $G-C$.
		\end{proof}
	\end{proof}
	Now, we put it all together. Suppose that the bad events we previously identified do not happen. In particular, we know that with probability at least $7/8$ we have:
	\[
	\bigg\lvert \hat{\eta} - \eta \bigg \rvert \leq \frac{\epsilon}{4}
	\]
	Additionally, we also know that with probability at least $7/8$ it is the case that:
	\begin{multline*}
	\bigg\lvert \frac{\big \lvert \left\{x: \substack{f_{\text{support}}(x)=1 \text{ and}\\ \text{$x$ is good and covered} }
		\right\}  \big \rvert}{2^n} - \frac{ \big\lvert \left\{ x: \substack{f_{\text{support}}(x)=1 \text{ and }\\\text{$x$ is good}}
		\right\}  \big\rvert}{2^n} \bigg \rvert
	\leq\\
	\frac{\epsilon}{4} \cdot \frac{\lvert \{ x:f_{\text{support}}(x)=1 \text{ and $x$ is good}
		\}  \rvert}{2^n}
	\end{multline*}
	By union bound, the probability that none of this bad events happens is at least $3/4$, which we will henceforth assume. Using the inequalities above together with the fact that the fraction of bad points is at most $\epsilon/2$ we get: 
	\begin{multline*}
	\bigg\lvert \hat{\eta}-\frac{\lvert \{ x: f_{\text{support}}(x)=1
		\}  \rvert}{2^n} \bigg \rvert
	\leq
	\bigg\lvert \hat{\eta} - \eta \bigg \rvert
	+
	\bigg\lvert \eta - \frac{\lvert \{ x: f_{\text{support}}(x)=1
		\}  \rvert}{2^n} \bigg \rvert
	=\\
	\bigg\lvert \hat{\eta} - \eta \bigg \rvert
	+
	\bigg\lvert \frac{\lvert \{ x: f_{\text{support}}(x)=1 \text{ and $x$ is covered}
		\}  \rvert}{2^n} - \frac{\lvert \{ x: f_{\text{support}}(x)=1
		\}  \rvert}{2^n} \bigg \rvert
	\leq
	\bigg\lvert \hat{\eta} - \eta \bigg \rvert
	+\\
	\bigg\lvert \frac{\big\lvert \left\{ x: \substack{f_{\text{support}}(x)=1 \text{ and}\\ \text{$x$ is good and covered}}
		\right\}  \big\rvert}{2^n} - \frac{\big\lvert \left\{ x: \substack{ f_{\text{support}}(x)=1 \text{ and }\\ \text{$x$ is good}}
		\right\}  \big\rvert}{2^n} \bigg \rvert
	+\\
	\frac{\lvert \{ x: f_{\text{support}}(x)=1 \text{ and $x$ is bad}
		\}  \rvert}{2^n}
	\leq
	\frac{\epsilon}{4}+\frac{\epsilon}{4} \cdot \frac{\lvert \{ x: f_{\text{support}}(x)=1 \text{ and $x$ is good}
		\}  \rvert}{2^n} + \frac{\epsilon}{2}
	\leq\\
	\frac{\epsilon}{4}+\frac{\epsilon}{4}+\frac{\epsilon}{2}=\epsilon
	\end{multline*}
	This completes the proof of correctness.
\end{proof}
\section{A lower bound on tolerant testing of uniformity}
In this section we prove a sample complexity lower bound on the problem of tolerantly testing the uniformity of an unknown monotone probability distribution over $\{0,1\}^n$: the task of distinguishing a distribution that is $o(1)$-close to uniform from a distribution that is sufficiently far from uniform. Recall the theorem:
\theoremDistanceToUniformLowerBound*
\begin{proof}
	A basic building block of our construction is the following:
	\begin{definition}
		For a member of the Boolean cube $x$, the \textbf{subcube distribution} $S_x$ is the probability distribution that picks $y$ uniformly, subject to $y \succeq x$.  
	\end{definition}
	All our distributions will be mixtures of such subcube distributions. For all the mixtures we will use, each subcube in the mixture is given the same weight. This method involving subcube distributions was used in \cite{rubinfeldservedio2009testing} to prove property testing lower bounds for monotone probability distributions.
	
	We construct $\Delta_{\text{Close}}$ to have only one member, which is equal to the uniform mixture of $S_x$ for all $\binom{n}{n^{0.5-0.01}}$ values of $x$ with Hamming weight $n^{0.5-0.01}$.
	
	We define a random member of $\Delta_{\text{Far}}$ to be the uniform mixture of $\frac{1}{2} 2^{n^{0.5-0.01}}$ subcube distributions $S_{x_j}$, where each of the $x_j$ is picked randomly among all the members of the Boolean cube with Hamming weight $n^{0.5-0.01}$. 
	
	We show that any member of $\Delta_{\text{Far}}$ is sufficiently far from uniform by upper-bounding the size of its support (i.e. the number of elements that have non-zero probability). Each of the subcube distributions has a support size of $2^{n-n^{0.5-0.01}}$. The support size of a mixture of distributions is at most the sum of the supports sizes of the respective distributions. Therefore, the support size of a member of $\Delta_{\text{Far}}$ is at most: \[2^{n-n^{0.5-0.01}} \cdot \frac{1}{2} 2^{n^{0.5-0.01}}=\frac{1}{2} 2^{n}\] This is sufficient to conclude that any member of $\Delta_{\text{Far}}$ is $1/2$-far from uniform.
	
	A random member $D_1$ of $\Delta_{\text{Far}}$ and the sole member $D_2$ of $\Delta_{\text{Close}}$ cannot be reliably distinguished using only $o\left(2^{\frac{n^{0.5-0.01}}{2}}\right)$ samples. This follows by the argument used in \cite{rubinfeldservedio2009testing}: Because of the number of samples, with probability at least $0.99$, the samples drawn from a random distribution from $D_1$ will all be from different subcube distributions. Also with probability at least $0.99$, this will also be true for the sole distribution of $D_2$. If both of these things happen (which is the case with probability at least $0.98$), the samples will be statistically indistinguishable. Thus, no tester can distinguish between $D_1$ and $D_2$ with an advantage greater than $0.02$.
	
	Finally, we need to prove that $D_2$ is $o(1)$-close to the uniform distribution. Here, the proof goes as follows. Both $D_2$ and the uniform distribution are symmetric with respect to a change of indices. This implies that the distance between these probability distributions equals to the distance between random variables $R_2$ and $R_1$, where $R_1$ is distributed as the Hamming weight of a random sample from $D_2$, whereas $R_2$ is distributed as the Hamming weight of uniformly random element of the Boolean cube. It is not hard to see that $R_1$ and $R_2$ are distributed according to binomial distributions with slightly different parameters. Now, the problem is equivalent to proving that the two following probability distributions are $o(1)$-close in total variation distance:
	\begin{itemize}
		\item A sum of $n$ i.i.d. uniform random variables from $\{0,1\}$.
		\item A sum of $n-n^{0.5-0.01}$ i.i.d. uniform random variables from $\{0,1\}$.
	\end{itemize}
	
	It is convenient to first bound the variation distance between 1) the sum of $k$ i.i.d. uniform random variables from $\{0,1\}$ and 2) $k+1$ i.i.d. uniform random variables from $\{0,1\}$, where $k$. We write the total variation distance as:
	\begin{multline*}
	\frac{1}{2^k}
	+
	\sum_{i=1}^{k-1} \bigg \lvert   \frac{1}{2^k}\binom{k}{i} -  \frac{1}{2^{k-1}} \binom{k-1}{i} \bigg\rvert
	=
	\frac{1}{2^k}
	\left(1+\sum_{i=1}^{k-1} \bigg \lvert   \binom{k-1}{i} -  \binom{k-1}{i-1} \bigg\rvert
	=
	\right)=\\
	\frac{1}{2^k}
	\left(1+\sum_{i=1}^{(k-1)/2}\left(  \binom{k-1}{i} -  \binom{k-1}{i-1}\right)
	+
	\sum_{i=(k-1)/2}^{k-1}\left(  \binom{k-1}{i-1} - \binom{k-1}{i} \right)
	\right)
	= \\
	\frac{1}{2^k}
	\left(
	1+\binom{k-1}{(k-1)/2}-1
	+\binom{k-1}{(k-1)/2}-1
	\right)
	=O \left(\frac{1}{\sqrt{k}}\right)
	\end{multline*}
	We telescoped the sums, and used the inequality that for all $k$, we have that $\binom{k}{k/2} \leq O\left(\frac{2^k}{\sqrt{k}} \right)$. For simplicity, we assumed above that $k-1$ is even, the odd case can be handled analogously.
	Thus, we have an upper bound of $O(1/\sqrt{k})$ on the total variation distance. 
	
	Using this, together with the triangle inequality for total variation distance, we bound the variation distance between 1) the sum of $n$ i.i.d. uniform random variables from $\{0,1\}$ and 2) the sum of $n-n^{0.5-0.01}$ i.i.d. uniform random variables from $\{0,1\}$ by \[O\left(\frac{n^{0.5-0.01}}{{n}^{0.5}}\right)=o(1)\].
	
	This finishes the proof.
\end{proof}
\bibliography{mybib}

\section{Appendix A}
\subsection{Verifying the conditions on $L_h$.}
\label{appendix subsection: verifying the conditions on L_h}
Recall that we defined $A$ and $L_h$ as follows:
\begin{itemize}
	\item $A:=\frac{1}{2n} \cdot e^{\frac{1}{2000} \cdot n^{1/5}}$
	\item For all $h \geq n/2$, we set $L_h:=\max\left(\log\left(2 n A \cdot \frac{\binom{n}{h}}{2^n} \right
	),0\right)$
	\item For all $h$, satisfying $n/2 > h \geq 9 \sqrt{n}$, we set:
	$L_h:= L_{n/2}=\log\left(2 n A \cdot \frac{\binom{n}{n/2}}{2^n} \right)$.
\end{itemize}

Here we prove that these values of $A$ and $L_h$ satisfy the following four conditions:
\begin{itemize}
	\item[a)] As a function of $h$, $L_h$ is non-increasing.
	\item[b)] For all $h$, we have that $L_h \leq 9 \sqrt{n}$.
	\item[c)] \[\frac{1}{2^n}
	\cdot 
	\sum_{h=9 \sqrt{n}}^n
	\binom{n}{h}
	\cdot
	\frac{A}{2^{L_{h}}} \leq \frac{1}{2}\]
	\item[d)] \[\sum_{h=9 \sqrt{n}}^n L_h \cdot \left(\begin{cases}
	\frac{400}{n^{2.5}} &\text{ if $h \leq n/2-\sqrt{n \ln(n)}$} 
	\\ \frac{40000}{n} &\text{ if $n/2-\sqrt{n \ln(n)}< h < n/2+\sqrt{n}$} \\ 40000 \cdot \left(\frac{h-n/2}{n}\right)^2 &\text{ if $h \geq n/2+\sqrt{n}$}\end{cases} \right)\leq \frac{\epsilon^2}{20000}\]
\end{itemize}

We will need the following standard fact can be proven, for example, by comparing $\sum_{i=0}^N i^k$ and $\int_{i=0}^N i^k \: di$:
\begin{fact}
	\label{fact: sum of kth powers}
	For any positive constant $k$ and for sufficiently large $n$, it is the case that:
	$
	\sum_{i=0}^n i^k=(1+o(1))\:\frac{n^{k+1}}{k+1}
	$.
\end{fact}

The truth of conditions (a) and (b) follows immediately by inspection. In fact a statement stronger than (b) is the case: for sufficiently large $n$ we have $L_h\leq \log(n \cdot A) \leq 2 \cdot n^{1/5} $. Regarding condition (c), we have:
\begin{multline*}
\frac{1}{2^n}
\cdot
\sum_{h=9 \sqrt{n}}^n
\binom{n}{h}
\cdot
\frac{A}{2^{L_{h}}}
=\\
\frac{A}{2^n}
\left(
\sum_{h=9 \sqrt{n}}^{n/2}
\binom{n}{h}
\frac{1}{2nA} \cdot \frac{2^n}{\binom{n}{n/2}}
+
\sum_{h=n/2}^n
\binom{n}{h}
\cdot
\min \left(\frac{1}{2nA} \cdot \frac{2^n}{\binom{n}{h}} ,1\right)
\right)
\leq
\sum_{h=9 \sqrt{n}}^n \frac{1}{2n} \leq \frac{1}{2}
\end{multline*}
Finally, recall that for all $h$, we have $L_h \leq  2 \cdot n^{1/5}$. For sufficiently large $n$, we have:
\begin{multline}
\label{equation verifying condition d 1}
\sum_{h=9 \sqrt{n}}^n L_h \cdot \left(\begin{cases}
\frac{400}{n^{2.5}} &\text{ if $h \leq n/2-\sqrt{n \ln(n)}$} 
\\ \frac{40000}{n} &\text{ if $n/2-\sqrt{n \ln(n)}< h < n/2+\sqrt{n}$} \\ 40000 \cdot \left(\frac{h-n/2}{n}\right)^2 &\text{ if $h \geq n/2+\sqrt{n}$}\end{cases} \right)\
\leq \\
\sum_{h=9 \sqrt{n}}^{n/2-\sqrt{n \ln n}} 2 \cdot n^{1/5} \cdot \frac{400}{n^{2.5}}
+
\sum_{h=n/2-\sqrt{n \ln n}}^{n+\sqrt{n}} 2 \cdot n^{1/5} \cdot \frac{40000}{n}
+
\sum_{h=n/2+\sqrt{n}}^n 40000 \cdot \left(\frac{h-n/2}{n}\right)^2 \cdot L_h
\leq \\
\frac{\epsilon^2}{40000}+\sum_{h=n/2+\sqrt{n}}^n 40000 \cdot \left(\frac{h-n/2}{n}\right)^2 \cdot L_h
\end{multline}
We now bound the last term using Hoeffding's inequality to bound the value of $\binom{n}{h}$, making a change of variables with $i:=\frac{n-h}{2}$ and then using Fact \ref{fact: sum of kth powers} to bound the resulting summation. Precisely, we have the following chain of inequalities (some of which are only true for sufficiently large $n$):
\begin{multline*}
\sum_{h=n/2+\sqrt{n}}^n 40000 \cdot \left(\frac{h-n/2}{n}\right)^2 \cdot L_h
\leq\\
\sum_{h=n/2}^n 40000 \cdot \left(\frac{h-n/2}{n}\right)^2 \cdot \: \max\left(\log\left(2n A \cdot \frac{\binom{n}{h}}{2^n} \right
),0\right)
\leq \\
\sum_{h=n/2}^n 40000 \cdot \left(\frac{h-n/2}{n}\right)^2 \cdot \; \max\left(\log\left(2n A \cdot \exp \left(-2\frac{(h-n/2)^2}{n} \right) \right
),0\right)
=\\\sum_{i=0}^{\sqrt{\frac{n}{2} \ln(2nA)}} 
\frac{40000}{\ln 2} \cdot \left( \frac{i}{n} \right)^2
\left( 
\ln (2nA) - 2 \cdot \frac{i^2}{n}
\right)
=\\
\frac{40000}{\ln 2} \cdot
\left((1+o(1)) \: \frac{\left(\sqrt{\frac{n}{2} \ln(2nA)}\right)^3}{3n^2} \:  
\ln (2nA) - (1+o(1)) \:2 \cdot \frac{\left(\sqrt{\frac{n}{2} \ln(2nA)}\right)^5}{5n^3}
\right)
\end{multline*}

Finally, simplifying and substituting the value of $A$ we get:
\[
\sum_{h=n/2+\sqrt{n}}^n 40000 \cdot \left(\frac{h-n/2}{n}\right)^2 \cdot L_h
\leq
(1+o(1)) \cdot
\frac{40000}{\ln 2} \cdot \frac{\sqrt{2}}{30 \sqrt{n}}
\left( \ln(2nA) \right)^{5/2}
\leq\frac{\epsilon^2}{40000}
\]
Condition (d) is verified by combining Equation \ref{equation verifying condition d 1} with the equation above.

\subsection{Proof of Claim \ref{claim: bound on sum of binomials divided by binomial}}
\label{appendix subsection: proof of claim bound on sum of binomials divided by binomial}
Here we prove that for all sufficiently large $n$, for all $h$, satisfying $0 \leq h \leq n$, it is the case that:
\[
\frac{\binom{n}{h}}{\sum_{j \geq h}^n \binom{n}{j}}
\leq
\left(\begin{cases}
\frac{2}{n^2} &\text{ if $h \leq n/2-\sqrt{n \ln(n)}$}
\\ \frac{200}{\sqrt{n}} &\text{ if $n/2-\sqrt{n \ln(n)}< h < n/2+\sqrt{n}$} \\ 
200 \cdot \frac{h-n/2}{n} &\text{ if $h \geq n/2+\sqrt{n}$}\end{cases} \right)
\]

We first handle the case when $h \geq n/2+\sqrt{n}$. If, furthermore, $h > 11n/20$, then it is sufficient to prove that $\frac{\binom{n}{h}}{\sum_{j \geq h}^n \binom{n}{j}} \leq 10$, which is trivially true. Thus, we now assume that $h \leq 11n/20$

It is the case that:
\begin{equation}
\label{equation: ratio of consequtive binomials.} \frac{\binom{n}{k}}{\binom{n}{k-1}}=\frac{1-\frac{k-n/2-1}{n/2}}{1+\frac{k-n/2}{n/2}}
\end{equation}
Therefore, we can write:
\[
\frac{\sum_{j \geq h}^n \binom{n}{j}}{\binom{n}{h}}
=
\sum_{j=h}^n  \prod_{k=h+1}^j
\frac{1-\frac{k-n/2-1}{n/2}}{1+\frac{k-n/2}{n/2}}
\]
Since $n/2 +\sqrt{n} \leq h \leq 11n/20$, for sufficiently large $n$ we have that $h+\frac{1}{4} \cdot \frac{n}{h-n/2}\leq n$. Using this, we can truncate the sum above, and then lower-bound the result by the product of the smallest summand with the total number of summands, getting:
\[
\frac{\sum_{j \geq h}^n \binom{n}{j}}{\binom{n}{h}}
\geq
\sum_{j=h}^{h+\frac{1}{4} \cdot \frac{n}{h-n/2}}  \prod_{k=h+1}^j
\frac{1-\frac{k-n/2-1}{n/2}}{1+\frac{k-n/2}{n/2}}
\geq
\frac{1}{4} \cdot
\frac{n}{h-n/2} \cdot \prod_{k=h+1}^{h+\frac{1}{4} \cdot\frac{n}{h-n/2}}
\frac{1-\frac{k-n/2-1}{n/2}}{1+\frac{k-n/2}{n/2}}
\]
Now, we analogously lower-bound the product by lower-bounding each of the factors, and then use the fact that since $h \geq n/2+\sqrt{n}$, it is the case that $ \frac{n}{h-n/2} \leq h-n/2$. We get:
\begin{multline*}
\frac{\sum_{j \geq h}^n \binom{n}{j}}{\binom{n}{h}} \geq
\frac{1}{4} \cdot
\frac{n}{h-n/2} \cdot \left(
\frac{1-\frac{h+\frac{1}{4} \cdot\frac{n}{h-n/2}-n/2-1}{n/2}}{1+\frac{h+\frac{1}{4} \cdot \frac{n}{h-n/2}-n/2}{n/2}}\right)^{
	\frac{1}{4} \cdot \frac{n}{h-n/2}}
\geq\\ \frac{1}{4} \cdot
\frac{n}{h-n/2} \cdot \left(
\frac{1-1.25\frac{h-n/2}{n/2}}{1+1.25\frac{h-n/2}{n/2}}\right)^{\frac{1}{4} \cdot \frac{n}{h-n/2}}
\end{multline*}
Finally, we use the fact that for all $w$ between zero and one we have that $\frac{1}{1+w}=1-w+w^2-...\geq 1-w$. We get:
\[
\frac{\sum_{j \geq h}^n \binom{n}{j}}{\binom{n}{h}} \geq
\frac{1}{4} \cdot
\frac{n}{h-n/2} \cdot \left(
1-1.25\frac{h-n/2}{n/2}\right)^{\frac{1}{2} \cdot \frac{2n}{h-n/2}}
\]

Now, recall that for any value $w$ between zero and one, we have that $\ln(1-w)=-\sum_{i=1}^{\infty} \frac{w^i}{i} \geq -\sum_{i=1}^{\infty} w^i=-\frac{w}{1-w}$. Using this, and recalling that $h\leq 11n/20$, we get that:
\[
\ln
\left(
1-1.25\frac{h-n/2}{n/2}\right)
\geq
-\frac{1.25\frac{h-n/2}{n/2}}{1-1.25\frac{h-n/2}{n/2}}
\geq
-\frac{1.25\frac{h-n/2}{n/2}}{1-1.25\frac{11n/20-n/2}{n/2}}
=-\frac{20}{7} \frac{h-n/2}{n}
\]
Combining the two previous equations together we get:
\[
\frac{\sum_{j \geq h}^n \binom{n}{j}}{\binom{n}{h}} \geq
\frac{1}{4} \cdot \frac{n}{h-n/2} \cdot \exp
\left(
-\frac{20}{7} \frac{h-n/2}{n}
\cdot
\frac{1}{2} \cdot
\frac{n}{h-n/2}
\right)
\geq
\frac{1}{200} \cdot \frac{n}{h-n/2}
\]
This completes the proof in the case $h \geq n/2+\sqrt{n}$.

Given our bound in the range $h \geq n/2+\sqrt{n}$, to show the desired bound in the range $n/2-\sqrt{n \ln(n)} < h < n/2 + \sqrt{n}$ it is sufficient to show that $\frac{\binom{n}{h}}{\sum_{j \geq h}^n \binom{n}{j}}$ is non-decreasing, as a function of $h$. If $h < n/2$, this follows immediately, because, as a function of $h$, the numerator is non-decreasing, whereas the denominator is decreasing. If $h \geq n/2$, then using Equation \ref{equation: ratio of consequtive binomials.}, we get: 
\begin{multline*}
\frac{\sum_{j \geq h+1}^n \binom{n}{j}}{\binom{n}{h+1}}
=
\frac{1+\frac{h+1-n/2}{n/2}}{1-\frac{h-n/2}{n/2}}
\cdot
\frac{1}{\binom{n}{h}}
\cdot
\sum_{j \geq h+1}^n \frac{1-\frac{j-n/2-1}{n/2}}{1+\frac{j-n/2}{n/2}} \binom{n}{j-1}
\leq\\
\frac{1}{\binom{n}{h}}
\cdot
\sum_{j \geq h+1}^n \binom{n}{j-1}
\leq
\frac{1}{\binom{n}{h}}
\cdot
\sum_{j \geq h+1}^{n+1} \binom{n}{j-1}
=
\frac{\sum_{j \geq h}^n \binom{n}{j}}{\binom{n}{h}}
\end{multline*}

Which implies that $\frac{\binom{n}{h+1}}{\sum_{j \geq h+1}^n \binom{n}{j}} \geq \frac{\binom{n}{h}}{\sum_{j \geq h}^n \binom{n}{j}}$

Finally, for the range $h \leq n/2-\sqrt{n \ln(n)}$ we can use Hoeffding's bound:
\[
\frac{\binom{n}{h}}{\sum_{j \geq h}^n \binom{n}{j}}
\leq
\frac{\binom{n}{h}}{\frac{1}{2} \cdot 2^n}
\leq
2 \cdot \Pr_{x \sim \{0,1\}^n}\left[x \leq n/2-\sqrt{n \ln(n)}\right]
\leq
2
\cdot \exp \left(-2\ln n \right)
=
\frac{2}{n^2}
\]
\subsection{Proof fo Claim \ref{claim: phi hat is close to phi}}
\label{appendix subsection: phi hat is close to phi}

Recall that:
\[N \myeq \frac{2^n}{A} \cdot \frac{192}{\epsilon^2} \cdot (n+9 \sqrt{n}+4)
\]
Our algorithm for learning a monotone probability distribution
drew $N$ samples from the probability distribution $\rho$ and the resulting multiset of samples was denoted as $S$.
For all $x$ in $\{0,1\}^n$, if $||x|| < 9 \sqrt{n}$, we set $\hat{\phi}(x)=0$, otherwise we set:
\[
\hat{\phi}(x)
:=
\frac{1}{2^{\left \lfloor L_{||x||} \right \rfloor}} \cdot \frac{\max_{y \text{ s.t. } y \preceq x \text{ and } ||y||-||x|| = \left \lfloor L_{||x||} \right \rfloor}
	\bigg \lvert \bigg \{z \in S: y \preceq z \preceq x \bigg\} \bigg \rvert}{N}
\]
Where $L_h$ is a specific value associated to each value of $h$. We also defined for all $x$ with $x \geq 9 \sqrt{n}$ the value:
\begin{multline*}
\phi(x)
\myeq
\frac{1}{2^{\left \lfloor L_{||x||}\right \rfloor}} \cdot \max_{y \text{ s.t. } y \preceq x \text{ and } ||x||-||y|| = {\left \lfloor L_{||x||}\right \rfloor}} \text{  }
\Pr_{z \sim \rho}[ y \preceq z \preceq x ]
=\\
\frac{1}{2^{\left \lfloor L_{||x||}\right \rfloor}} \: \cdot \max_{y \text{ s.t. } y \preceq x \text{ and } ||x||-||y|| = {\left \lfloor L_{||x||}\right \rfloor }} \text{  }
\sum_{z \text{ s.t. } y \preceq z \preceq x} \rho(z)
\end{multline*}

Here we prove that if it is the case that $\frac{1}{2^n}
\cdot
\sum_{h=9 \sqrt{n}}^n
\binom{n}{h}
\cdot
\frac{A}{2^{L_{h}}} \leq \frac{1}{2}$, 
then, with probability at least $7/8$, it is the case that:
\[
\sum_{x \in \{0,1\}^n \text{ s.t. } 9 \sqrt{n} \leq ||x||}
\bigg \lvert \hat{\phi}(x)-\phi(x) \bigg \rvert
\leq
\frac{\epsilon}{4}
\]

We claim that for any pair $(x,y)$, such that $\phi$ is defined on $x$ and $||x||-||y|| = \left \lfloor L_{||x||} \right \rfloor$, with probability at least $1-\frac{1}{8\cdot 2^n \cdot n^{9 \sqrt{n}}}$ the following holds:
\begin{equation}
\label{equation: chernoff plus hoeffding learning}
\bigg \lvert
\Pr_{z \sim \rho}[ y \preceq z \preceq x ]
-
\frac{
	\lvert  \{z \in S: y \preceq z \preceq x \}  \rvert}{N}
\bigg \rvert
\leq 
\frac{\epsilon}{8} \cdot \max \left(\frac{A}{2^n}
, \Pr_{z \sim \rho}[ y \preceq z \preceq x ] \right)
\end{equation}
We use Chernoff's bound to prove this as follows. Denote by $q$ the value $\Pr_{z \sim \rho}[ y \preceq z \preceq x ]$. If $q \geq \frac{A}{2^n}$ then by Chernoff's bound we have:
\begin{multline*}
\Pr \left[
\bigg \lvert
\lvert  \{z \in S: y \preceq z \preceq x \}  \rvert
-qN
\bigg \rvert \geq \frac{\epsilon}{8} qN
\right]
\leq\\
2\exp\left(-\frac{1}{3} \left(\frac{\epsilon}{8} \right)^2 qN \right)
\leq
2\exp\left(-\frac{1}{3} \left(\frac{\epsilon}{8} \right)^2 \frac{A}{2^n} \cdot N \right)
=
\frac{1}{8\cdot 2^n \cdot n^{9 \sqrt{n}}}
\end{multline*}
Otherwise, if we have $q < \frac{A}{2^n}$, then by Chernoff's bound:
\begin{multline*}
\Pr \left[
\bigg \lvert
\lvert  \{z \in S: y \preceq z \preceq x \}  \rvert
-qN
\bigg \rvert \geq \frac{\epsilon}{8} \cdot \frac{A}{2^n} \cdot N
\right]
\leq
2\exp\left(-\frac{1}{3} \left(\frac{\epsilon}{8} \cdot \frac{A}{2^n} \cdot \frac{1}{q} \right)^2 q N \right)
\leq\\
2\exp\left(-\frac{1}{3} \left(\frac{\epsilon}{8} \right)^2 \frac{A}{2^n} \cdot N \right)
=
\frac{1}{8\cdot 2^n \cdot n^{9 \sqrt{n}}}
\end{multline*}

Now, by taking a union bound, it follows that with probability $7/8$ for all such pairs $(x,y)$ Equation \ref{equation: chernoff plus hoeffding learning} will be the case. Recalling the definition of $\phi$, for all $x$ on which $\phi$ is defined it then will be the case that:
\[
\bigg \lvert
\hat{\phi}(x)-\phi(x)
\bigg \rvert
\leq 
\frac{\epsilon}{8} \cdot
\max \left( \frac{1}{2^{\left \lfloor L_{||x||} \right \rfloor}} \frac{A}{2^n}
, \phi(x) \right)
\]
Now, we sum this for all $x$ in the domain of $\phi$ and use the fact that $\lfloor L_h \rfloor \geq L_h+1$, and then use that $\frac{1}{2^n}
\cdot
\sum_{h=9 \sqrt{n}}^n
\binom{n}{h}
\cdot
\frac{A}{2^{L_{h}}} \leq \frac{1}{2}$. We get:
\begin{multline*}
\sum_{x \in \{0,1\}^n \text{ s.t. } 9 \sqrt{n} \leq ||x||}
\bigg \lvert \hat{\phi}(x)-\phi(x) \bigg \rvert
\leq\\
\frac{\epsilon}{8}
\cdot
\left(
\frac{1}{2^n}
\cdot
\sum_{x \in \{0,1\}^n \text{ s.t. } 9 \sqrt{n} \leq ||x||}
\frac{A}{2^{ \left \lfloor L_{||x||}  \right \rfloor}}
+
\sum_{x \in \{0,1\}^n \text{ s.t. } 9 \sqrt{n} \leq ||x||} \phi(x)
\right)
=\\
\frac{\epsilon}{8}
\cdot
\left(
\frac{1}{2^n}
\cdot
\sum_{h=9 \sqrt{n}}^n
\binom{n}{h}
\cdot
\frac{A}{2^{ \lfloor L_{h}  \rfloor}}
+
\sum_{x \in \{0,1\}^n \text{ s.t. } 9 \sqrt{n} \leq ||x||} \phi(x)
\right)
\leq\\
\frac{\epsilon}{8}
\cdot
\left(
2 \cdot
\frac{1}{2^n}
\cdot
\sum_{h=9 \sqrt{n}}^n
\binom{n}{h}
\cdot
\frac{A}{2^{L_{h}}}
+
\sum_{x \in \{0,1\}^n \text{ s.t. } 9 \sqrt{n} \leq ||x||} \phi(x)
\right)
\leq\\
\frac{\epsilon}{8}
\cdot
\left(
1
+
\sum_{x \in \{0,1\}^n \text{ s.t. } 9 \sqrt{n} \leq ||x||} \rho(x)
\right)
\leq
\frac{\epsilon}{4}
\end{multline*}

\end{document}